\documentclass[journal]{IEEEtran}
\usepackage{cite}
\ifCLASSINFOpdf
\else
\fi
\usepackage{amsthm}
\usepackage{amssymb}
\usepackage{color}
\usepackage{graphics}

\usepackage{units}
\usepackage{tabularx}
\usepackage{amsmath}
\usepackage{algorithm}
\usepackage[noend]{algorithmic}
\ifCLASSOPTIONcompsoc
  \usepackage[caption=false,font=normalsize,labelfont=sf,textfont=sf]{subfig}
\else
  \usepackage[caption=false,font=footnotesize]{subfig}
\fi
\usepackage{makecell}

\usepackage{float}
 \newcommand{\N}{\mathbb{N}}
 
 \newcommand{\R}{\mathbb{R}}
 \newcommand{\Z}{\mathbb{Z}}
 \newcommand{\X}{\mathcal{X}}
 \newcommand{\U}{\mathcal{U}}
 \newcommand{\diag}{\mathrm{diag}}

\usepackage{url}
\usepackage[stretch=10]{microtype}
\usepackage[capitalize]{cleveref}

\crefname{equation}{}{} % Default to eqref for equations
\crefname{section}{Sec.}{Sec.}
\crefname{enumi}{}{}
\crefname{assumption}{Assumption}{Assumptions}

\newtheorem{thm}{Theorem}
\crefname{thm}{Theorem}{Theorems}
\newtheorem{cor}[thm]{Corollary}
\crefname{cor}{Corollary}{Corollaries}
\newtheorem{lem}[thm]{Lemma}
\crefname{lem}{Lemma}{Lemmas}

\newtheorem{assum}[thm]{Assumption}
\crefname{assum}{Assumption}{Assumptions}

\newtheorem{defn}[thm]{Definition}
\crefname{defn}{Definition}{Definitions}
\newtheorem{exmp}[thm]{Example}
\newtheorem{rem}[thm]{Remark}

\crefname{alg}{Algorithm}{Algorithms}

\begin{document}
\title{Learning-based Model Predictive Control for Safe Exploration and Reinforcement Learning}

\author{Torsten~Koller$^*$,
        Felix~Berkenkamp$^*$,
        Matteo Turchetta,
        Joschka~Boedecker
        and Andreas Krause% <-this % stops a space
\thanks{$^*$ T. Koller and F. Berkenkamp contributed equally to this work.}
\thanks{T. Koller and J. Boedecker are with the Neurorobotics Group, University of Freiburg, Freiburg, Germany (email: \{kollert, jboedeck\}@informatik.uni-freiburg.de\}).}
\thanks{F. Berkenkamp, M. Turchetta and A. Krause are with Learning and Adaptive Systems Group, ETH Zurich, Zurich, Switzerland (email: \{befelix, matteotu, krausea\}@inf.ethz.ch\}).}% <-this % stops a space
\thanks{This work was supported by SNSF grant {200020\_159557}, a fellowship within the FITweltweit program of the German Academic Exchange Service (DAAD), an Open Philantropy Project AI fellowship, the Max Planck ETH Center for Learning Systems, and the BrainLinks-BrainTools  Cluster  of  Excellence  (grant  number  EXC1086). T. Koller is funded through the State Graduate Funding Program of Baden-Wuerttemberg}% <-this % stops a space
% <-this % stops a space
}

\maketitle

\begin{abstract}
Reinforcement learning has been successfully used to solve difficult tasks in complex unknown environments. However, these methods typically do not provide any safety guarantees during the learning process. This is particularly problematic, since reinforcement learning agents actively explore their environment. This prevents their use in safety-critical, real-world applications.
In this paper, we present a learning-based model predictive control scheme that provides high-probability safety guarantees throughout the learning process.
Based on a reliable statistical model, we construct provably accurate confidence intervals on predicted trajectories. Unlike previous approaches, we allow for input-dependent uncertainties.
Based on these reliable predictions, we guarantee that trajectories satisfy safety constraints. Moreover, we use a terminal set constraint to recursively guarantee the existence of safe control actions at every iteration.
We evaluate the resulting algorithm to safely explore the dynamics of an inverted pendulum and to solve a reinforcement learning task on a cart-pole system with safety constraints.
\end{abstract}

%\begin{IEEEkeywords}
%?, ?, ?
%\end{IEEEkeywords}

% \fi
\IEEEpeerreviewmaketitle

% !TEX root = ../root.tex

\section{Introduction}
\label{intro}
In model-based reinforcement learning (RL,\cite{Sutton1998Reinforcement}), we aim to learn the dynamics of an unknown system from data and, based on the model, derive a policy that optimizes the long-term behavior of the system. In order to be successful, these methods need to find an \textit{efficient} trade-off between \textit{exploitation}, i.e. finding and executing optimal policies based on the current knowledge of the system and \textit{exploration}, where the system traverses regions of the state-space that are unknown or uncertain to collect observations that improve the model estimation. 
In real-world safety-critical systems, we further need to restrict exploration to actions that are safe to perform by satisfying state and control constraints at all times.
In contrast, current approaches often use exploration schemes that lead to unpredictable and possibly unsafe behavior of the system, e.g., based on the \textit{optimism in the face of uncertainty} principle \cite{Xie2016Modelbased}. Consequently, such approaches have limited applicability to real-world systems. 

In this paper, we present a safe learning-based model predictive control (MPC) scheme for nonlinear systems with state-dependent uncertainty. The method guarantees the existence of feasible \textit{return trajectories} to a safe region of the state space at every time step with high probability. Specifically, we estimate the system uncertainty based on observations through statistical modelling techniques. Based on the statistical model, we derive two reliable uncertainty propagation techniques that allow for state-dependent uncertainty. We show that the resulting MPC scheme provides high probability safety guarantees and combine it with exploration schemes from RL. We show that our algorithm can safely learn about the system dynamics (system identification) when it is combined with an exploration scheme that aims to reduce the uncertainty about the model through safe exploration. Moreover, we apply our algorithm to a reinforcement learning task where the algorithm only aims to collect \textit{task-relevant} data in order to maximize expected performance on a RL task. In our experiments, we empirically show the importance of both, planning a safety trajectory in avoiding crashes and planning a performance trajectory in solving the task and how the design choices of our algorithms can improve the performance in a RL task. 

Preliminary results on one of the uncertainty propagation scheme for safe control, but without RL, were shown in \cite{Koller2018Learningbased}. This paper provides additional information, proposes an additional uncertainty propagation scheme, shows how the algorithm can be combined with RL, and provides a more thorough experimental evaluation.

\begin{figure}[t]
\includegraphics{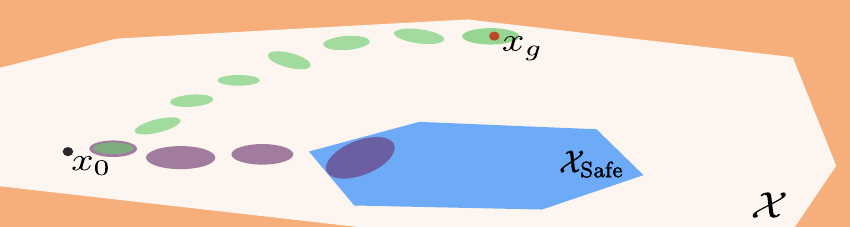}
\caption{Simultaneous planning of a performance trajectory (green ellipsoids) using an approximate uncertainty propagation technique and a safety trajectory (purple ellipsoids) using the proposed robust multi-step ahead prediction technique. While the performance trajectory optimizes an estimation of the expected long-term utility (e.g. distance to goal state $x_g$) of applying a control input to the system at $x_0$, the safety trajectory guarantees that the same input could eventually return the system to the safe set $\mathcal{X}_\mathrm{Safe}$ without violating the safety constraints $\mathcal{X}$ and control constraints $\mathcal{U}$.}
\label{main:mpc_rl:fig:safety_performance}
\end{figure}

\paragraph*{Outline}
We provide a brief introduction to our work and outline our contributions in \cref{intro}, relating them to previous work in \cref{related_work}. The problem statement is given in \cref{ps} and the necessary background for our research is provided in \cref{back}. In \cref{multi_step}, we derive our multi-step ahead prediction technique. In \cref{safempc}, we then use this technique to derive a MPC scheme and prove its safety. \cref{safe_rl} discusses how we can efficiently use this for exploration and to solve a safe model-based RL task. In \cref{practice}, we briefly address practical application issues and in \cref{exp}, we report our experimental results in safe exploration and safe RL on an inverted pendulum and a cart-pole system, respectively. We conclude our paper with a discussion in \cref{discuss}.
% !TEX root = ../root.tex

\section{Related work}
\label{related_work}
One area that considers safety guarantees in partially unknown systems is robust MPC. These methods iteratively optimize the performance along finite-length trajectories at each time step, based on a known model that incorporates uncertainties and disturbances acting on the system \cite{rawlings2009model}. In a constrained robust MPC setting, we optimize these local trajectories under additional state and control constraints. Safety is typically defined in terms of recursive feasibility and robust constraint satisfaction.
%, where feasibility of the underlying MPC problem is obtained for all time steps, given an initial feasible solution, guaranteeing that the system constraints are satisfied at all times.
In~\cite{Sadraddini2016Provably}, this definition is used to safely control urban traffic flow, while~\cite{Carson2013Robust} guarantees safety by switching between a standard and a safety mode. However, these methods are conservative, since they do not update the model.

In contrast, \emph{learning-based MPC} approaches (LBMPC) adapt their models online based on observations of the system~\cite{Aswani2013Provably}. This allows the controller to improve over time, given limited prior knowledge of the system. Theoretical safety guarantees in LBMPC are established in \cite{Aswani2013Provably, Wabersich2018Linear}. Both approaches compute \textit{tubes} around trajectories of a known nominal linear model, which accounts for all disturbances. While the former approach requires disturbances to be bounded in an a-priori specified  polytope, the latter relies on sampling to approximate model errors. Robust constraint satisfaction and closed-loop stability for LBMPC with nominal linear models and state-dependent uncertainty is shown in \cite{Soloperto2018LearningBased}. However, the method relies on extensive offline computations. In \cite{Rosolia2019SampleBased}, instead of updating a model of the system, samples are used to grow a safe set in a robust MPC setting over time.   %In \cite{Maiworm2018Stability}, input-to-state stability of a MPC approach is shown 

Another area that considers learning for control is \emph{model-based RL}, where we learn a model from observations of our system and use it to find a policy for a given task based on an external cost signal \cite{Sutton1998Reinforcement}. One statistical model that is commonly used in model-based RL are \emph{Gaussian processes} (GP, \cite{Rasmussen2006Gaussian}). In RL, the GP uncertainty estimates are typically used to optimize the expected long-term cost of a feedback policy by propagating the state-dependent uncertainties in the forward predictions of the GP model over multiple time steps. However, trajectories have complex dependencies on states and unbounded stochastic uncertainties, which requires approximation techniques to get tractable solutions. The expected cost under the forward propagated approximate state distributions are then optimized w.r.t. parameterized or MPC-based policies \cite{Deisenroth2011PILCO,Cao2017Gaussian,Kamthe2017DataEfficient,Boedecker2014Approximate}.

In contrast to MPC, where we optimize finite-length trajectories based on a pre-specified system model, in RL we typically aim to find an infinite horizon optimal policy under unknown dynamics of our system. Hence, enforcing hard constraints in RL is challenging. In \cite{Garcia2012Safe}, a safe exploration heuristic for RL is introduced that uses the proximity to a finite set of safe states generated from pre-defined exploration policies as the safety criterion.
In \cite{Dalal2018Safe}, a safety layer is introduced that acts on top of an RL agent's possibly unsafe actions to prevent immediate negative consequences, but does not deal with negative long-term consequences of an action. 
Safety based on constrained Markov decision processes \cite{Altman1999Constrained} is considered in  \cite{Chow2019Lyapunovbased,Achiam2017Constrained}, where the RL agent additionally receives a constraint cost signal and the task is to keep the expected sum of constraint costs below a given threshold. The issue with approaches based on this notion of \textit{expected safety} is that we can still have arbitrary large and frequent violations of our constraints as we do not account for higher moments of the accumulated constraint cost distribution. Control-theoretic safety properties such as Lyapunov stability or robust constraint satisfaction have received little attention~\cite{Ernst2009Reinforcement}. In~\cite{Berkenkamp2017Safe}, safety is guaranteed by optimizing parametric policies under stability constraints, while~\cite{Akametalu2014Reachabilitybased} guarantees safety in terms of constraint satisfaction through reachability analysis. In \cite{Wabersich2018Safea}, a safety filter is introduced that employs robust MPC to correct for possibly unsafe actions of an RL agent. While the authors show safety in terms of constraint satisfaction, the RL agent and the safety mechanism are fundamentally detached from each other. This could potentially lead to suboptimal and possibly oscillatory behavior, where the RL agent constantly tries to approach the constraints.
In \cite{Chen2018Approximating}, methods from reinforcement learning are used to approximately solve the constrained linear quadratic regulator (LQR), but the system is assumed to be known and the approach is limited to linear systems. In \cite{Dean2018Safely}, a learning-based technique for the constrained LQR for unknown linear dynamics with sub-optimality bounds and statistical learning guarantees is proposed. 
In GP-based RL, safety is often considered through probabilistic \textit{chance constraints} \cite{Hewing2017Cautious,Jain2018Learning,Ostafew2016Robust} based on approximate uncertainty propagation. While often empirically successful, these approaches do not theoretically guarantee safety of the underlying system, since constraints are typically only enforced over finite horizons and without incorporating the approximation errors of the multi-step ahead predictions.

% !TEX root = ../root.tex

%\section{Notation}
%\label{not}
%In this section, we introduce the notation and give an overview over variables and some of the definitions used throughout this paper.
%Sets are denoted by upper-case caligraphic letters, e.g. $\mathcal{S} \subset \R $, i.e. a subset of the set of real numbers.
%Vector-valued variables are denoted by lower-case letters, e.g. a $\in \R^n$. Matrices are given by upper-case letters, e.g. $M \in \R^{n \times m}$. 
%A list of variables, functions and constants is given in \cref{not:table:definitions}

% !TEX root = ../root.tex

\section{Problem Statement}
\label{ps}
We consider a deterministic, discrete-time  dynamical system
\begin{equation}
\label{ps:eq:system}
x_{t+1} = f(x_t,u_t) = \underbrace{h(x_t,u_t)}_{\text{prior model}} + \underbrace{g(x_t,u_t)}_{\text{unknown error}},
\end{equation}
where $x_t \in \R^{p}$ is the state and $u_t \in \R^{q}$ is the control input to the system at time step $t\in \N$. We assume that we have access to a known, Lipschitz continuous prior model  $h(x_t,u_t)$ and that the \textit{a priori} unknown model-error $g(x_t,u_t)$ is Lipschitz continuous as well. The prior model $h$ could be based on an idealized first principles physics model that may be inaccurately identified. By collecting transition data, we want to reduce these model uncertainties and \textit{learn} the unknown model-error $g$ from the observations of our system. In order to learn about~$g$, we use a statistical model with mean and variance given by $(\mu_n(x_t,u_t),\Sigma_n(x_t,u_t))$. This provides us with a point-wise approximation of $g$, given by $\mu_n$, and corresponding \textit{input-dependent} uncertainty estimates $\Sigma_n(x_t,u_t) = \diag([\sigma_{n,1}^2(x_t,u_t),..,\sigma_{n,p}^2(x_t,u_t)])$, where the subscript $n \in \N$ denotes the number of observations we have collected from our system so far. As we collect more observations from the system during operation, we want to learn and to improve this estimate over time. 

We assume that the system is subject to polytopic state and control constraints
\begin{align}
\label{ps:ass:state_constraits}
\X = \{x \in \R^{p} | H_x x \leq h_x, \, h_x \in \R^{m_x} \}, \\
\label{ps:ass:ctrl_constraits}
\U = \{u \in \R^{q} | H_u u \leq h_u, \, h_u \in \R^{m_u} \},
\end{align}
which are given by the intersection of $m_x, m_u \in \N$ individual half-spaces, respectively. As a running example, we consider an autonomous driving scenario, where the state region could correspond to a highway lane or race track and the control constraints could represent the physical limits on acceleration and steering angle of the car.

In order to provide safety guarantees, we need reliable estimates of the model-error inside the operating region $\mathcal{X} \times \mathcal{U}$ of our system. Restricted to this region, we make the assumption that, with high probability jointly throughout the operating time of our system, at a finite number of inputs in each time step, values of our unknown but deterministic function  $g$ are contained in the confidence intervals of our statistical model.
\begin{assum}
\label{prob:ass:reliable_statistical_model}
 For all $n \in \N$ and a fixed sample size $T \in \N$, pick arbitrary sets of inputs $\mathcal{D}_n = \{z_1,..,z_T\} \subset \mathcal{X} \times \mathcal{U}$. For any confidence level $\delta \in (0,1]$, we have access to scaling factors $\beta_{n,T} > 0$ such that uniformly with probability greater than $1-\delta$, we have that $\forall \, n \in \N,\, 1\leq j \leq p, 1\leq k \leq T$: 
\begin{equation}
    |\mu_{n,j}(z_k) - g_j(z_k)| \leq \beta_{n,T} \cdot \sigma_{n,j}(z_k).
\end{equation}
\end{assum}
Intuitively, this means that we can trust our model to return well-calibrated estimates. For notational convenience, we drop the subscripts of $\beta_{n,T}$ and refer to it as $\beta$ in the remainder of the paper.
In general, providing these kind of guarantees is impossible for arbitrary functions~$g$. We provide further analysis in \cref{back} on how we can guarantee this sort of reliability when using a Gaussian process (GP) statistical model to approximate a certain class of deterministic functions. 

Lastly, we assume access to a backup controller that guarantees that we never violate our safety constraints once we enter a safe subset of the state space. In the autonomous driving example, this could be a simple linear controller that stabilizes the car in a small region in the center of the lane.
\begin{assum}
\label{ps:ass:safe_set}
We are given a controller $\pi_{\mathrm{safe}}(\cdot)$ and a polytopic safe region
\begin{equation}
\X_{\mathrm{safe}} := \{ x\in \R^{p} | H_s x \leq h_s, h_s \in \R^{m_{s}} \} \subseteq \X,
\end{equation}
at the intersection of ${m_s \in \N}$ individual half-spaces. We denote with
$x_{t+1} = f_{\pi_{\mathrm{safe}}} = f(x_t,\pi_{\mathrm{safe}}(x_t))$ the closed-loop system under $\pi_{\mathrm{safe}}$ and assume for some arbitrary $k \in \N$ that
\begin{equation}
   x_k \in \mathcal{X}_{\mathrm{safe}} \Rightarrow  f_{\pi_{\mathrm{safe}}}(x_t) \in \mathcal{X}, \quad  \forall t \geq k.
\end{equation}
\end{assum}

This assumption allows us to gather initial data from the system even in the presence of significant model errors, since the system remains safe under the controller~$\pi_\mathrm{safe}$. Moreover, we can still guarantee constraint satisfaction asymptotically outside of~$\X_\mathrm{safe}$, if we can show that a finite sequence of control inputs eventually steers the system back to the safe set. In \cite{Koller2018Learningbased}, we stated a stricter assumption that required $\X_\mathrm{safe}$ to be robust control positive invariant (RCPI) under $\pi_{\mathrm{safe}}$, i.e. the system is not allowed to leave the safe set once it enters it. Our relaxed definition only requires the safe set to be a \emph{subset} of a RCPI set that is contained in the set of admissible states~\cite{Wabersich2018Linear}.
Given a controller $\pi$, ideally we want to enforce the state- and control constraints at every time step,
\begin{equation}
\label{ps:eq:safety_deterministic}
\forall t \in \N: f_\pi(x_t) \in \mathcal{X}, \, \pi(x_t) \in \mathcal{U}.
\end{equation}
Apart from $\pi_{\mathrm{safe}}$, which trivially and conservatively fulfills this if $x_0 \in  \X_{\mathrm{safe}}$, it is in general impossible to design a controller that enforces
\eqref{ps:eq:safety_deterministic} without additional assumptions. This is due to the fact that we only have access to noisy observations of our system. Instead, we slightly relax this requirement to \textit{safety with high probability} throughout its operation time.
\begin{defn}
\label{ps:def:epsilon_safety}
Let $\pi: \R^{p} \to \R^{q}$ be a controller for \eqref{ps:eq:system} with the corresponding closed-loop system $f_\pi$.
Let $x_0 \in \mathcal{X}_{\mathrm{safe}}$ and $\delta \in (0,1]$. A system is $\delta-$safe under the controller $\pi$ iff:
\begin{equation}
\Pr\left[ \, \forall t \in \N:  f_\pi(x_t) \in \mathcal{X}, \, \pi(x_t) \in \mathcal{U}  \right] \geq 1-\delta.
\end{equation}
\end{defn}
As a comparison, in chance constrained MPC we typically get \textit{per time step} high probability constraint satisfaction. Informally speaking, there we have definitions of the form $\forall t \in \N: \Pr[\text{No collision at time }t] \geq 1-\delta$. In contrast, \cref{ps:def:epsilon_safety} requires high probability safety guarantees that are \emph{independent} from the duration of system operation.

Aside from these mandatory safety requirements, we are provided with a mission objective and we want our controller to improve its performance in solving a given task as we learn about our system trough observations. In this paper, we are particularly interested in two intimately related tasks. Firstly, we are interested in safely and efficiently identifying the system by exploring the safe region of the state space.
Secondly, we want to find a controller that solves a RL task by minimizing the accumulated sum of costs our system receives in every state while still remaining safe at all times. This naturally entails the former task, namely,  a mechanism that can safely explore the state space while trying to solve the task. The next section elaborates on this.
% !TEX root = ../root.tex

\section{Background}
\label{back}
In this section, we provide a brief introduction to model-based RL and introduce the necessary background on GPs and set-theoretic properties of ellipsoids that we use to model our system and perform multi-step ahead predictions.

\subsection{Model-based Reinforcement Learning}
In reinforcement learning (RL), an agent interacts with an unknown environment according to a policy $\pi: \R^p \rightarrow \R^q$ by applying an action $u_t = \pi(x_t) \in \R^q$ in a state $x_t \in \R^p$. The environment defines a probability distribution over the next state $X_{t+1} \sim p(x_{t+1}|x_t,u_t)$. After every time step the agent receives a scalar cost $c(x_t,u_t) \in \R$. 
The performance of an agent is then measured through the \textit{expected cumulative cost}
\begin{equation}
\label{eq:back:model_based_rl:exp_cost}
    V^\pi(x) = \sum_{t=0}^{\infty}\mathbb{E}_{x_t \sim X_t}[ \gamma^t c(x_t,\pi(x_t))| x_0 = x], \, x \in \R^p,
\end{equation}
where the \textit{discount-factor} $\gamma \in [0,1)$ guarantees convergence of the infinite summation. In our special case of a deterministic environment, we can remove the expectation over state trajectories. 
The goal of a reinforcement problem is to find 
\begin{equation}
\label{eq:back:model_based_rl:optimization}
    \pi^* = \arg\min_{\pi} \, V^\pi(x),
\end{equation}
the policy that minimizes the expected cost of interacting with its environment.

Approaches to solving these RL tasks can generally be subsumed under \textit{model-based} and \textit{model-free} RL algorithms. We are interested in the former, where we try to \textit{learn} an approximate model of the environment from experience samples  $(x_t,u_t,x_{t+1})$ that allows us to directly approximate $V^\pi(x)$ in equation \eqref{eq:back:model_based_rl:exp_cost} for an agent $\pi$. This is typically done by replacing the infinite sum in \eqref{eq:back:model_based_rl:exp_cost} with a finite horizon look-ahead and using the model to approximate the expectation over trajectories. Depending on the class of policies under consideration, various ways of optimizing \eqref{eq:back:model_based_rl:optimization} are proposed in the literature, including MPC \cite{Boedecker2014Approximate,Kamthe2017DataEfficient}.

\subsection{Gaussian Processes and RKHS}
One way to learn the unknown model-error $g$ from data is by using a GP model.
A $\mathcal{GP}(m,k)$ is a distribution over functions, which is fully specified through a mean function~$m: \R^{d} \to \R$ and a \textit{covariance function} $k: \R^d \times \R^d \to \R$, where $d = p + q$. Given a set of $n$ noisy observations\footnote{For notational simplicity, we only outline the equations for $p=1$, i.e. a single-output system.} $\hat{f}_i = f(z_i) + w_i, \, w_i \sim \mathcal{N}(0,\lambda^2),\, z_i = (x_i,u_i), \,i = 1,\dots,n, \, \lambda \in \R$, we choose a zero-mean prior on $g$ as $m \equiv 0$ and regard the differences $y_n = [\hat{f}_1 - h(z_1),\dots,\hat{f}_n - h(z_n)]^T$ between prior model $h$ and observed system response at input locations $Z = [z_1,..,z_n]^T$. The posterior distribution at~$z$ is then given as a Gaussian $\mathcal{N}(\mu_n(z),\sigma_n^2(z))$
with mean and variance
\begin{align}
\label{found:gp:pred_mu}
\mu_n(z) &= k_n(z)^\mathrm{T} [K_n+\lambda^2 I_n]^{-1}y_n, \\
\label{found:gp:pred_var}
\sigma^2_n(z) &= k(z,z) - k_n(z)^\mathrm{T} [K_n + \lambda^2 I_n]^{-1}k_n(z),
\end{align}
where $[K_n]_{ij} = k(z_i,z_j), [k_n(z)]_j = k(z,z_j)$, and $I_n$ is the $n-$dimensional identity matrix.
In the case of multiple outputs $p > 1$, we model each output dimension with an independent GP, $\mathcal{GP}(m_j,k_j), j = 1,..,p$. We then redefine \eqref{found:gp:pred_mu} and \eqref{found:gp:pred_var} as $\mu_n(\cdot) = (\mu_{n,1}(\cdot),..,\mu_{n,n_x}(\cdot))$ and $\sigma_n(\cdot) = (\sigma_{n,1}(\cdot),..,\sigma_{n,n_x}(\cdot))$ corresponding to the predictive mean and variance functions of the individual models.

\subsubsection{RKHS}
We consider the space of functions $\mathcal{H}_k$ that is defined through the completion of all possible posterior mean functions $\eqref{found:gp:pred_mu}$ of a $\mathcal{GP}(0,k)$, i.e.
\begin{equation}
\label{back:eq:gp:rkhs}
    \mathcal{H}_k = \left\{\sum_{l=0}^\infty \alpha_l k(z,z_l) | \, \alpha_l \in \R, z,z_l \in \R^{d}  \right\}.
\end{equation}
This class of functions is well-behaved in the sense that they form a reproducing kernel Hilbert space (RKHS, \cite{wahba1990spline}) equipped with an inner-product $\langle \cdot,\cdot \rangle_k$. The induced norm $||\cdot||_k^2$ is a measure of the \textit{complexity} of a function in $\mathcal{H}_k$.
We consider the case, where our model-error $g$ is indeed of the form $g(z) = \sum_{l=0}^\infty \alpha_l k(z,z_l), \, \alpha_l \in \R, z,z_l \in \R^{d}$, i.e. $g \in \mathcal{H}_k$. Using results from the statistical learning theory literature, we can directly derive useful properties from this assumption.
\begin{lem}
\label{gp:lemma:lipschitz}
\cite[Lemma 2]{Berkenkamp2016Safe}:
Let g have bounded RKHS norm $||g||_k \leq B_g$ induced by a continuously differentiable kernel $k$. Then $g$ is Lipschitz continuous.
\end{lem}
That is, the regularity assumption on our model-error and the smoothness assumption on the covariance function $k$ directly imply that the function $g$ is Lipschitz.

More importantly, the assumption $g \in \mathcal{H}_k$ guarantees that using GPs to model the unknown part of the system~\eqref{ps:eq:system}, provides us with reliable confidence intervals on the model-error $g$ as required in \cref{prob:ass:reliable_statistical_model}.
\begin{lem}
\label{gp:lemma:rkhs_confidence_ivals}
\cite[Lemma 2]{Berkenkamp2017Safe}: Assume $||g||_k \leq B_g$ and that measurements are corrupted by $\lambda$-sub-Gaussian noise. Let $\beta_n = B_g + 4\lambda \sqrt{\gamma_n + 1 + \ln(1/\delta)}$, where $\gamma_n$ is the information capacity associated with the kernel $k$. Then with probability at least $1-\delta$, with $\delta \in (0,1)$, we have for all $ 1\leq j \leq n_x,\, z \in \mathcal{X} \times \mathcal{U}$ that $|\mu_{n-1,j}(z) - g_j(z)| \leq \beta_n \cdot \sigma_{n-1,j}(z)$.
\end{lem}

In combination with the prior model $h(z)$, this allows us to construct reliable confidence intervals around the true dynamics of the system \cref{ps:eq:system}. The scaling $\beta_n$ depends on the number of data points~$n$ that we gather from the system through the information capacity,
$\gamma_n = \max_{\mathcal{A} \subset \tilde{\mathcal{Z}} , |A| = \tilde{n}} I(\tilde{g}_A;g), \,  \tilde{\mathcal{Z}} = \mathcal{X} \times \mathcal{U} \times \mathcal{I}, \, \tilde{n} = n \cdot p, \, \mathcal{I} = \{1,..,p\}$, i.e. the maximum mutual information $I(\tilde{g}_A,g)$ between a finite set of samples $A$ and the function $g$.
Exact evaluation of $\gamma_n$ is NP-hard in general, but it can be greedily over-approximated with bounded error and has sublinear dependence on $n$ for many commonly used kernels~\cite{Srinivas2010Gaussian}. \cref{gp:lemma:rkhs_confidence_ivals} is stronger than required by \cref{prob:ass:reliable_statistical_model}, since the confidence level $\delta$ holds jointly for all $z \in \mathcal{X} \times \mathcal{U}$ and not only for a finite set of $T$ samples. As a side note, we refer to a result in \cite{Srinivas2010Gaussian} showing that \cref{prob:ass:reliable_statistical_model} holds if our system transitions $x_{t+1}$ is represented by stochastic draws from the GP posterior \cref{found:gp:pred_mu,found:gp:pred_var}, with $z = (x_t, u_t)$.\footnote{Due to the difficulty of designing a safety controller and safe region that would fulfill \cref{ps:ass:safe_set} in such a stochastic system (see \cite{Vinogradska2016Stability} for a discussion), we nonetheless consider a deterministic system.}

\subsection{Ellipsoids}
We use ellipsoids to bound the uncertainty of our system when making multi-step ahead predictions.
Due to appealing geometric properties, ellipsoids are widely used in the robust control community to compute \textit{reachable sets} \cite{Filippova2017Ellipsoidal,asselborn2013control}. These sets intuitively provide an outer approximation on the next state of a system considering all possible realizations of uncertainties when applying a controller to the system at a given set-valued input.
We briefly review some of these properties and refer to \cite{Kurzhanskii1997Ellipsoidal} for an exhaustive introduction to ellipsoids and to the derivations for the following properties.

We use the basic definition of an ellipsoid,
\begin{equation}
\label{eq:found:ell:def}
\mathcal{E}(p,Q) := \{ x \in \R^n | (x-p)^\mathrm{T} Q^{-1}(x-p) \leq 1\},
\end{equation}
with  \textit{center} $p \in \R^n$ and a symmetric positive definite (s.p.d) \textit{shape matrix} $Q \in \R^{n \times n}$.
Ellipsoids are invariant under \textit{affine subspace transformations} such that, for $A \in \R^{n \times r}, r \leq n$ with full column rank and $b \in \R^r$, we have that
\begin{equation}
\label{multi_step:ell:affine_trafo}
A \cdot \mathcal{E}(p,Q) + b = \mathcal{E}(p+b,AQA^\mathrm{T}).
\end{equation}
The \textit{Minkowski sum} $\mathcal{E}(p_1,Q_1) \oplus \mathcal{E}(p_2,Q_2)$, i.e. the pointwise sum between two arbitrary ellipsoids, is in general not an ellipsoid anymore, but we have
\begin{equation}
\label{multi_step:ell:minkowski_ellipsoids}
\mathcal{E}(p_1,Q_1) \oplus \mathcal{E}(p_2,Q_2) \subset \mathcal{E}_c(\tilde{p},\tilde{Q}), 
\end{equation}
where $\tilde{p} = p_1 + p_2, \,\tilde{Q} = (1+c^{-1})Q_1 + (1+c)Q_2$ for all $c > 0$. Moreover, the minimizer of the trace of the resulting shape matrix is analytically given as $c = \sqrt{Tr(Q_1) / Tr(Q_2)}$. A particular problem that we encounter is finding the maximum distance $r$ to the center of an ellipsoid $\mathcal{E}(0,Q)$ under a special transformation, i.e.
\begin{align}
\label{found:ell:prob:rayleigh_variant}
r(Q,S) = \max_{x \in \mathcal{E}(p,Q)} ||S(x-p)||_2 =  \max_{s^\mathrm{T} Q^{-1}s \leq 1} s^\mathrm{T} S^\mathrm{T} S s,% \\
%&= \max_{x \in E(0,Q)} ||S x||^2_2
%&= \max_{ x \in R^{n_x} x^TQ^{-1}x ||S x||^2_2
\end{align}
where $S \in \R^{m \times n}$ with full column rank. This is a generalized eigenvalue problem of the pair $(Q,S^TS)$ and the optimizer is given as the square-root of the largest generalized eigenvalue.%, where $S = [I_n,K^T]$.

% !TEX root = ../root.tex

\section{Robust multi-step ahead predictions}
\label{multi_step}

In order to plan safe trajectories based on our statistical model, we need to reliably estimate the region of the state space that can be reached over multiple time steps under a sequence of control inputs. Based on \cref{prob:ass:reliable_statistical_model} and our prior model~$h(x_t,u_t)$, we directly obtain high-probability confidence intervals on~$f(x_t,u_t)$ uniformly for all $t \in \N$ given a single control input $u_t$. We extend this to over-approximate the system after a sequence of inputs $(u_t,u_{t+1},..,u_{T-1})$. The result is a sequence of set-valued confidence regions that contain the true trajectory of the system with high probability.

\subsection{One-step ahead predictions}
%We compute an ellipsoidal confidence region that contains the next state of the system with high probability when applying a control input, given that the current state is contained in an ellipsoid.
We derive a function that takes as an input an ellipsoidal subset of the state space and outputs a second ellipsoid providing a confidence region around the next state of the system under a specific control input, given that the current state is contained in the input ellipsoid.

\begin{figure}[t]
    \includegraphics{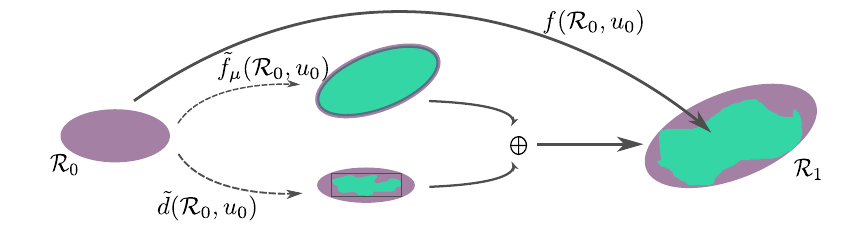}
    \caption{Decomposition of the over-approximated image of the system \cref{ps:eq:system} under an ellipsoidal input $R_0$. The exact, unknown image of $f$ (right, green area) is approximated by the linearized model $\tilde{f}_\mu$ (center, top) and the remainder term $\tilde{d}$, which accounts for the confidence interval and the linearization errors of the approximation (center, bottom). The resulting ellipsoid $R_1$ is given by the Minkowski sum of the two individual approximations.}
    \label{main:multi_step:fig:one_step_overapproximation}
\end{figure}

In order to approximate the system, first we linearize our prior model $h(x_t,u_t)$ and use the affine transformation property \eqref{multi_step:ell:affine_trafo} to compute the ellipsoidal next state of the linearized prior model. Next, we approximate the unknown model-error $g(x_t,u_t)$ using the confidence intervals of our statistical model. We propose a locally constant approximation of $g$ in \cref{robust_multistep:one_step:locally_const}, while in \cref{robust_multistep:one_step:locally_linear}, we locally linearize our statistical model.  We finally apply Lipschitz arguments to outer-bound the approximation errors. We sum up these individual approximations, which result in an ellipsoidal approximation of the next state of the system. This is illustrated in~\cref{main:multi_step:fig:one_step_overapproximation}. We formally derive the necessary equations in the following paragraphs. The reader may choose to skip the technical details of these approximations, which result in ~\cref{main:multi_step:lemma:one_step}. To guide the interested reader through the derivations, we provide a list of important variables functions and constants in \cref{tab:variables}.

\begin{table}[t]
\label{not:table:definitions}
\caption{List of important variables, functions and constants.} 
\begin{tabularx}{\linewidth}{c|c|X}
Variable & Type & Definition \\
\hline
$f$ & Function& Dynamical system under consideration.\\
$h$ & Function& Known part of the system $f$.\\
$g$ & Function& Unknown part of the system $f$.\\
$x_t$ & Vector& State of the system at time $t$.  \\
$u_t$ & Vector& Control input to the system at time $t$.\\
$\mu_n$ & Function & {Predictive mean of a statistical model with $n$ observations.}\\
$\sigma_n^2$ & Function& {Predictive variance of a statistical model with $n$ observations.}\\
$\beta$ & Scalar & {Scaling factor for the confidence intervals of a statistical model.} \\
$\mathcal{R}$ & Set& {Ellipsoidal set of states $\mathcal{R} = \mathcal{E}(p,Q)$, with center $p$ and shape matrix $Q$.} \\
%$\mathcal{X}_{\mathrm{Safe}}$ & Set & Polytopic safe set \\
%$\pi_{\mathrm{Safe}}$ & Function & Initial safe controller \
$J_\phi$ & Matrix & {Jacobian matrix of a function $\phi$ with $J_\phi = [A_\phi, B_\phi]$, where $A_\phi, B_\phi$ are the Jacobians w.r.t. the state and control inputs, respectively.} \\
$L_\phi$ & Vector & {Vector with Lipschitz constants for each output of a function $\phi$.} \\
$P_{\phi}^{\bar{z}}$ & Function & {Taylor approximation of a function $\phi$ around linearization point $\bar{z}$.} \\
$\tilde{m}(\mathcal{R},\pi_t)$& Function & {Function taking ellipsoidal set of states $\mathcal{R}$ and affine feedback controller $\pi_t$. Outputs ellipsoidal over-approximation of next system state. }
%\bottomrule
\end{tabularx}
\label{tab:variables}
\end{table}

\subsubsection{Predictions via locally constant model approximation}
\label{robust_multistep:one_step:locally_const}
We first regard the system $f$ in \eqref{ps:eq:system}
for a single input vector $z = (x,u), \,f(z) = h(z) + g(z)$.
We linearly approximate $f$ around $\bar{z} = (\bar{x},\bar{u})$ via
\begin{equation}
\label{main:multi_step:eq:taylor}
f(z) \approx h(\bar{z}) + J_{h}(\bar{z}) (z-\bar{z}) + g(\bar{z}) = \tilde{f}(z),
%f(z) \approx \underbrace{h(\bar{z}) + J_{h}(\bar{z}) (z-\bar{z}) + g(\bar{z})}_{ = \tilde{f}(z)},
\end{equation}
where $J_{h}(\bar{z}) = [A_h,B_h]$ is the Jacobian of $h$ at $\bar{z}$.

Next, we use the Lagrangian remainder theorem~\cite{Breiman1993Deterministic} on the linearization of $h$ and apply a continuity argument on our locally constant approximation of $g$. This results in an upper-bound on the approximation error,
\begin{equation}
\label{main:eq:linearization_approx}
|f_j(z) - \tilde{f}_j(z) | \leq \frac{L_{\nabla h,j}}{2} ||z - \bar{z}||_2^2 + L_g ||z-\bar{z}||_2,
\end{equation}
where~$f_j(z)$ is the $j$th component of~$f$, $ 1\leq j \leq p$, $L_{\nabla h,j}$ is the Lipschitz constant of the gradient $\nabla h_j$, and $L_g$ is the Lipschitz constant of~$g$, which exists by \cref{gp:lemma:lipschitz}.

The function~$\tilde{f}$ depends on the unknown model error~$g$. We approximate~$g$ with the statistical model, $\mu(\bar{z}) \approx g(\bar{z})$. From \cref{prob:ass:reliable_statistical_model} we have
\begin{equation}
\label{main:eq:rkhs_bound}
|g_j(\bar{z})-\mu_{j}(\bar{z})| \leq \beta \sigma_{j}(\bar{z}), \, 1\leq j \leq p,
\end{equation}
with high probability.
We combine \eqref{main:eq:linearization_approx} and \eqref{main:eq:rkhs_bound} to obtain
\begin{equation}
\label{main:eq:errorbound_deterministic}
| f_j(z) - \tilde{f}_{\mu,j}(z)| \leq \beta \sigma(\bar{z}) + \frac{L_{\nabla h,j}}{2} ||z - \bar{z}||_2^2 + L_g ||z-\bar{z}||_2,
\end{equation}
where $1 \leq j \leq p$ and
%\begin{equation}$
$
\tilde{f}_\mu(z) = h(\bar{z}) + J_{h}(\bar{z}) (z-\bar{z}) + \mu_n(\bar{z}).
$
We can interpret  \eqref{main:eq:errorbound_deterministic} as the edges of the confidence hyper-rectangle
\begin{equation}
\label{main:eq:model:deterministic}
\tilde{m}(z)  = \tilde{f}_\mu(z) \pm \bigg[ \beta\sigma(\bar{z})+\frac{L_{\nabla h}}{2} ||z - \bar{z}||_2^2 + L_g ||z-\bar{z}||_2 \bigg],
\end{equation}
where $L_{\nabla h} = [L_{\nabla h,1},..,L_{\nabla h,p}]$ and we use the shorthand notation $a \pm b := [a_1 \pm b_1] \times [a_{p} \pm b_{p}],\, a,b \in \R^{p}$.

We are now ready to compute a confidence region based on an ellipsoidal state $\mathcal{R} = \mathcal{E}(p,Q) \subset \R^{p}$ and a fixed input $u \in \R^{q}$, by over-approximating the output of the system $f(\mathcal{R},u) = \{ f(x,u)| x \in R\}$ for all inputs contained in an ellipsoid $\mathcal{R}$. Here, we choose $p$ as the linearization center of the state and choose $\bar{u} = u$, i.e. $\bar{z} = (p,u)$. Since the function $\tilde{f}_\mu$ is affine, we can make use of \eqref{multi_step:ell:affine_trafo} to compute
\begin{equation}
\label{multi_step:eq:linear_part:affine_trafo}
\tilde{f}_\mu(\mathcal{R},u) = \mathcal{E}(h(\bar{z}) + \mu(\bar{z}),A_hQA_h^\mathrm{T}),
\end{equation}
which results in an ellipsoid.
This is visualized in~\cref{main:multi_step:fig:one_step_overapproximation} by the upper ellipsoid in the center.
To upper-bound the confidence hyper-rectangle on the right hand side of \eqref{main:eq:model:deterministic}, we upper-bound the term~$\| z - \bar{z} \|_2$ by
% \begin{align}
% \label{main:eq:remainder_overapproximation_setinput}
% &\max\limits_{\substack{z(x) = (x,u), \\ x \in R}}\tilde{d}(z) =
%  &\beta_n \sigma_{n-1}(\bar{z}) + L_{\nabla h} \frac{l^2(R,u)}{2}  + L_gl(R,u), \notag \\
% % \end{equation}
%   &\textnormal{where~}
% % \begin{equation}
% l(R,u) = \max\limits_{\substack{z(x) = (x,u), \\ x \in R}} ||z(x) - \bar{z}||_2.
% \end{align}
\begin{equation}
\label{main:eq:remainder_overapproximation_setinput}
l(\mathcal{R},u) := \max\limits_{\substack{z(x) = (x,u), \\ x \in \mathcal{R}}} ||z(x) - \bar{z}||_2,
\end{equation}
which leads to
\begin{equation}
\label{main:eq:d_overapprox_rectangle}
\tilde{d}(\mathcal{R}, u)  := \beta\sigma(\bar{z}) + L_{\nabla h} l^2(\mathcal{R},u) / 2 + L_g l(R,u).
\end{equation}
Due to our choice of $z, \bar{z}$, we have that $||z(x) - \bar{z}||_2 = ||x-p||_2$ and  we can use \eqref{found:ell:prob:rayleigh_variant} to get
$l(\mathcal{R},u) = r(Q,I_{p}),$ which corresponds to the largest eigenvalue of $Q^{-1}$.
Using \eqref{main:eq:remainder_overapproximation_setinput}, we can now over-approximate the right side of \eqref{main:eq:model:deterministic} for inputs $R$ by an ellipsoid
\begin{equation}
0 \pm \tilde{d}(\mathcal{R},u) \subset \mathcal{E}(0,Q_{\tilde{d}}(\mathcal{R},u)),
\end{equation}
where we obtain $Q_{\tilde{d}}(\mathcal{R},u)$  by over-approximating the hyper-rectangle~$\tilde{d}(\mathcal{R}, u)$ with the ellipsoid~$\mathcal{E}(0,Q_{\tilde{d}}(\mathcal{R},u))$ through $a \pm b \subset \mathcal{E}(a,\sqrt{p}\cdot \mathrm{diag}([b_1,..,b_{p}])), \, \forall a,b \in \R^{p}$. This is illustrated in \cref{main:multi_step:fig:one_step_overapproximation} by the lower ellipsoid in the center.
Combining the previous results, we can compute the final over-approximation using~\eqref{multi_step:ell:minkowski_ellipsoids},
\begin{equation}
\label{main:multi_step:eq:one_step_ahead_set}
\mathcal{R}_+ = \tilde{m}(\mathcal{R},u) = \tilde{f}_\mu(R,u) \oplus  \mathcal{E}(0,Q_{\tilde{d}}(\mathcal{R},u)).
\end{equation}
Since we carefully incorporated all approximation errors and extended the confidence intervals around our model predictions to set-valued inputs, we get the following generalization of the reliability \cref{prob:ass:reliable_statistical_model} for single inputs. 

\begin{lem}
\label{main:multi_step:lemma:one_step}
Let $\delta \in (0,1]$, choose $\beta$ and $T$ as in \cref{prob:ass:reliable_statistical_model} and $D_n = \{(x_{1,n},u_{1,n}),..,(x_{T,n},u_{T,n})\} \subset \mathcal{X} \times \mathcal{U}$ for each $n \in \N$. Then, with probability greater than $1-\delta$, we have that:
\begin{equation}
 x_{k,n} \in \mathcal{R}_{k,n} \Rightarrow f(x_{k,n},u_{k,n}) \in \tilde{m}(\mathcal{R}_{k,n},u_{k,n}),
\end{equation}
uniformly for all $n \in \N, \, 1 \leq k \leq T,\, R_{k,n} = \mathcal{E}(p_{k,n},Q_{k,n}) \subset \mathcal{X}$.
\end{lem}
\begin{proof}
Define $m(x,u) = h(x,u)+\mu_n(x,u) \pm \beta \sigma_n(x,u)$. From \cref{gp:lemma:rkhs_confidence_ivals} we have $\forall \, \mathcal{R}_{k,n} \subset \mathcal{X},\, u_{k,n} \in \mathcal{U}$ with $x_{k,n} \in \mathcal{R}_{k,n}$ that, uniformly with high probability, $f(x_{k,n},u_{k,n}) \in m(x_{k,n},u_{k,n})$. Due to the over-approximations, we have $m(x_{k,n},u_{k,n}) \subset \tilde{m}(\mathcal{R}_{k,n},u_{k,n})$. 
\end{proof}

\cref{main:multi_step:lemma:one_step} allows us to compute confidence ellipsoid around the next state of the system, given that the current state of the system is known to be contained in an ellipsoidal confidence region.

\subsubsection{Predictions via model linearization}
\label{robust_multistep:one_step:locally_linear}
The previous derivations that lead to the reliable ellipsoidal one-step ahead predictions \eqref{main:multi_step:eq:one_step_ahead_set} use Lipschitz arguments on the unknown model-error $g$ to give an outer-bound bound on the linearization errors of our locally linear approximations. A different way to arrive at a reliable outer-bound, that can potentially reduce the conservatism of the proposed uncertainty propagation technique, is by using a first-order Taylor approximation of the statistical model $\mu_n$.
Following previous derivations, we can approximate our unknown model-error with $g(z) \approx \mu_n(\bar{z}) + A_{\mu_n} (x-\bar{x}) + B_{\mu_n} (u - \bar{u}) = P_{\mu_n}^{\bar{z}}(z)$, where $\bar{z} = (\bar{x},\bar{u})^{T}$ and $J_{\mu_n}(\bar{z}) = [A_{\mu_n},B_{\mu_n}]$ is the Jacobian of $\mu_n$ at $\bar{z}$.
Following \eqref{main:multi_step:eq:taylor} - \eqref{main:multi_step:eq:one_step_ahead_set}, we can get an alternative outer bound by making the approximation
\begin{equation}
    f(z) \approx  h(\bar{z}) + J_{h}(\bar{z}) (z-\bar{z}) + P_{\mu_n,j}^{\bar{z}}(z) = \tilde{f}_{P_{\mu_n}}(z).
\end{equation}
By adding and subtracting with $g(z)$ and $\mu(z)$, the approximation error between $f$ and $\tilde{f}_{P_\mu}$ can now be outer-bound via 
\begin{align}
    | f_j(z) -  \tilde{f}_{{P_{\mu_n}},j}(z)| &\leq  \frac{L_{\nabla h,j}}{2} ||z - \bar{z}||_2^2 \nonumber \\ 
    &+  |P_{\mu_n,j}^{\bar{z}}(z) - \mu_{n,j}(z)| + |\mu_{n,j}(z) - g_j(z)|,
\end{align}
for $z \in \mathcal{X}\times\mathcal{U}, 1 \leq j \leq p$.
We use the reliability \cref{prob:ass:reliable_statistical_model} on our statistical model combined with Lipschitz arguments to obtain
\begin{align}
     |P_{\mu_n,j}^{\bar{z}}(z) - \mu_{n,j}(z)| &\leq \frac{L_{\nabla \mu_n,j}}{2} ||z - \bar{z}||_2^2, \\
     |\mu_{n,j}(z) - g_j(z)| &\leq \beta(\sigma_n(\bar{z}) + L_{\sigma_n}||z - \bar{z}||_2),
\end{align}
where the second inequality holds uniformly
with high probability $\forall z \in D_n \subset \mathcal{X}\times\mathcal{U}, 1 \leq j \leq p, n \in \N$
%with high probability $\forall z \in  \mathcal{X}\times\mathcal{U}, 1 \leq j \leq p$. 

We regard this over-approximation for ellipsoidal inputs $\mathcal{R}=\mathcal{E}(p,Q)$. Closely following the derivations \eqref{multi_step:eq:linear_part:affine_trafo},\eqref{main:eq:d_overapprox_rectangle}, we can decompose our over-approximation into an affine part
\begin{equation}
    \tilde{f}_{P_\mu}(\mathcal{R},u) = \mathcal{E}(h(\bar{z}) + \mu(\bar{z}),(A_h+A_\mu)Q(A_h+A_\mu)^\mathrm{T})
\end{equation}
and combine the Lipschitz approximations into the hyper-rectangle
\begin{equation}
    0 \pm \beta_n(\sigma(\bar{z}) + L_\sigma l(R,u))+ \frac{L_{\nabla h}+L_{\nabla \mu}}{2} l^2(R,u),
\end{equation}
which results in the ellipsoidal over-approximation $\mathcal{E}(0,Q_{\tilde{d}_{P}}(\mathcal{R},u))$. We arrive at the alternative one-step ahead prediction 
\begin{equation}
\label{multi_step:linearization:eq:one_step_ahead_set}
\tilde{m}_{P_\mu}(\mathcal{R},u) = \tilde{f}_{P_\mu}(\mathcal{R},u) \oplus  \mathcal{E}(0,Q_{\tilde{d}_{T}}(\mathcal{R},u)).
\end{equation}
Applying the arguments of \cref{main:multi_step:lemma:one_step}, we can directly see that the function $\tilde{m}_{P_\mu}$ also provides reliable confidence ellipsoids for arbitrary ellipsoidal inputs $R \subset \mathcal{X}$.

Apart from being possibly less conservative in many applications, we note that the one-step ahead predictions \eqref{multi_step:linearization:eq:one_step_ahead_set} only require Lipschitz constants $L_{\nabla\mu}, L_{\nabla h}, L_{\sigma}$ of the \textit{known} functions $\nabla\mu, \nabla h, \sigma$. These could be identified using tools from the global optimization literature~\cite{Wood1996Estimation}.
\subsection{Multi-step ahead predictions}
We use the previous results to compute a sequence of ellipsoids that contain a trajectory of the system with high-probability, by iteratively applying the one-step ahead predictions~\eqref{main:multi_step:eq:one_step_ahead_set}. We note that the following line of arguments equivalently holds for the alternative one-step ahead prediction in~\eqref{multi_step:linearization:eq:one_step_ahead_set}.

Given an initial ellipsoid $\mathcal{R}_0 \subset \R^{p}$ and control input $u_t \in \R^q$, we iteratively compute confidence ellipsoids as
\begin{equation}
\label{multi_step:eq:t_step_delta_reach}
\mathcal{R}_{t+1} = \tilde{m}(\mathcal{R}_t,u_t).
\end{equation}
We can directly apply \cref{main:multi_step:lemma:one_step} to get the following result.
\begin{cor}
\label{multi_step:lemma:t_step_delta_reach}
Let $\delta \in (0,1]$ and choose $\beta$ and $T$ as in \cref{prob:ass:reliable_statistical_model}. For each $n \in \N$ choose $x_{0,n} \in \mathcal{R}_{0,n} \subset \mathcal{X},\, \{u_{0,n},..,u_{T-1,n}\} \subset \mathcal{U}$. Then the following holds jointly for all $n \in \N, 0 \leq t \leq T-1$ with probability at least $1-\delta$:
% \begin{equation}
$x_{t,n} \in R_{t,n}$,
% \end{equation}
where $(x_{t,n},u_{t,n}) \in \mathcal{X} \times \mathcal{U}$, $\mathcal{R}_{1,n},..,\mathcal{R}_{T,n} \subset \mathcal{X}$ is computed as in \eqref{multi_step:eq:t_step_delta_reach} and $\{x_{0,n},..,x_{T,n}\}$ are trajectories of the system under input sequences $\{u_{0,n},..,u_{T-1,n}\}$ for each $n \in \N$.
 \end{cor}
 \begin{proof}
 Since \cref{main:multi_step:lemma:one_step} holds for any set $\mathcal{D}_n, n \in \N$, we can choose $\mathcal{D}_n = \{(x_{0,n},u_{0,n}),..,(x_{T-1,n},u_{T-1,n})\}$ and so that starting from tarting in $x_{0,n} \in \mathcal{R}_{0,n}.$  we get with high probability that $x_{i,n} \in \mathcal{R}_{i,n} = \tilde{m}(\mathcal{R}_{i-1,n},u_{i-1,n})$ for all $i=1,..T$ with $n \in \N$.
\end{proof}

\cref{multi_step:lemma:t_step_delta_reach} guarantees that, with high probability, the system is always contained in the propagated ellipsoids \eqref{multi_step:eq:t_step_delta_reach}. However, this only holds if all confidence ellipsoids are inside the feasible state region $\mathcal{X}$. Thus, if we provide safety guarantees for these sequences of ellipsoids, we obtain high-probability safety guarantees for the system~\eqref{ps:eq:system}.
\subsection{Predictions under state-feedback control laws}
\label{safempc:affine_feedback}
When applying multi-step ahead predictions under a sequence of feed-forward inputs $u_t \in \mathcal{U}$, the individual sets of the corresponding reachability sequence can quickly grow unreasonably large. This is because these \textit{open loop} input sequences do not account for future control inputs that could correct deviations from the model predictions.
Hence, we extend \eqref{main:multi_step:eq:one_step_ahead_set} to
\textit{affine state-feedback control laws} of the form
\begin{equation}
\label{main:multi_sep:eq:feedback_ctrl}
u_{K,t}(x_t) := K_t(x_t-p_t) + k_t,
\end{equation}
where $K_t \in \R^{q \times p}$ is a feedback matrix and $k_t \in \R^{q}$ is the open-loop input. The parameter $p_t$ is determined through the center of the current ellipsoid $\mathcal{R}_t = \mathcal{E}(p_t,Q_t)$. Given an appropriate choice of $K_t$, the control law actively contracts the ellipsoids towards their center. This technique is commonly used in \textit{tube-based MPC}, to reduce the size of \textit{tubes} around a nominal trajectory of the system that incorporate uncertainties and disturbances \cite{rawlings2009model}.
Similar to the derivations \eqref{main:multi_step:eq:taylor}-\eqref{main:multi_step:eq:one_step_ahead_set}, we can compute the function $\tilde{m}$ for affine feedback controllers \eqref{main:multi_sep:eq:feedback_ctrl} and ellipsoids $R_t =\mathcal{E}(p_t,Q_t)$.
The resulting ellipsoid is
\begin{equation}
\label{main:multi_step:eq:over_approx_state_feedback}
\tilde{m}(\mathcal{R}_t,u_{K,t}) =  \mathcal{E}(h(\bar{z}_t)+\mu(\bar{z}_t),H_tQ_tH_t^\mathrm{T}) \oplus \mathcal{E}(0,Q_{\tilde{d}}(\mathcal{R}_t,u_{K,t})),
\end{equation}
where $\bar{z_t} = (p_t,k_t)^\mathrm{T}$ and $H_t = A_h + B_h K_t$. The set $\mathcal{E}(0,Q_{\tilde{d}}(\mathcal{R}_t,u_{K,t}))$ is obtained similarly to \eqref{main:eq:remainder_overapproximation_setinput} as the ellipsoidal over-approximation of
\begin{equation}
0 \pm [\beta_n \sigma(\bar{z}) + L_{\nabla h} \frac{l^2(\mathcal{R}_t,S_t)}{2}  + L_gl(\mathcal{R}_t,S_t)],
\end{equation}
with $S_t = [I_{n_x},K_t^\mathrm{T}]$ and $l(\mathcal{R}_t,S_t) = \max_{x \in \mathcal{R}_t}||S_t(z(x)-\bar{z_t})||_2 \,, z(x) = (x,u_{K_t}(x))^\mathrm{T}$. The theoretical results of \cref{main:multi_step:lemma:one_step} and \cref{multi_step:lemma:t_step_delta_reach} directly apply to the case of the uncertainty propagation technique \eqref{main:multi_step:eq:over_approx_state_feedback}. For the remainder of this paper, we assume $K_t$ is pre-specified, while $k_t$ is assumed to be a decision variable. For the sake of generality, we drop the subscript $K$ and the functional dependency on $x$ in $u_{K,t}(x)$ unless required and refer to \eqref{main:multi_sep:eq:feedback_ctrl} when writing $u_t$.

\subsection{Safety constraints}
The derived multi-step ahead prediction technique provides a sequence of ellipsoidal confidence regions around trajectories of the true system $f$ through \eqref{multi_step:eq:t_step_delta_reach}. We can guarantee that the system is safe by verifying that the computed confidence ellipsoids are contained inside the polytopic constraints \eqref{ps:ass:state_constraits} and \eqref{ps:ass:ctrl_constraits}.
That is, given a sequence of feedback controllers $u_{K,t}, \, t = {0,..,T-1}$ we need to verify
\begin{equation}
\label{multi_step:eq:constr:finite_horizon_safety}
\mathcal{R}_{t+1} \subset \mathcal{X}, \, u_t(\mathcal{R}_t) \subset \mathcal{U}, \, t=0,..,T-1,
\end{equation}
where $(R_0,..,R_T)$ is given through \eqref{multi_step:eq:t_step_delta_reach} and $u_t(R_t) := \{u_{K,t}(x) | x \in R_t\} $. Since our constraints are polytopes, we have that
% \begin{equation}
$\mathcal{X} = \bigcap_{i=1}^{m_x} \mathcal{X}_i$,
% \end{equation}
$\mathcal{X}_i = \{x \in \R^{p} |  [H_x]_{i,\cdot}  x - h_i^x\leq 0 \},$
where $[H_x]_{i,\cdot}$ is the $i$th row of $H^x$.
We can now formulate the state constraints through the condition $\mathcal{R}_t = \mathcal{E}(p_t,Q_t) \subset \mathcal{X}$ as $m_x$ individual constraints $R_t \subset \mathcal{X}_i,\, i= 1,..,m_x$, for which an analytical formulation exists~\cite{Hessem2002Closedloop},
\begin{equation}
\label{multi_step:ell:ell_polytope_state_constr}
[H_x]_{i,\cdot}p_t + \sqrt{[H_x]_{i,\cdot}Q_t[H_x]_{i,\cdot}^T}   \leq h^x_i, \,
\end{equation}
 $\forall i \in \{1,..,m_x\}$. Moreover, we can use the fact that $u_t$ is affine  in~$x$ to obtain $u_t(\mathcal{R}_t) = \mathcal{E}(k_t,K_tQ_t,K_t^T)$, using \eqref{multi_step:ell:affine_trafo}. The corresponding control constraint $u_t(\mathcal{R}_t) \subset \mathcal{U}$ is then equivalently given by
\begin{equation}
\label{multi_step:ell:ell_polytope_control_constr}
[H_u]_{i,\cdot}k_t + \sqrt{[H_u]_{i,\cdot}K_tQ_tK_t^\mathrm{T} [H_u]_{i,\cdot}^\mathrm{T}} \leq h^u_i, \, 
\end{equation}
$\forall i \in \{1,..,m_u\}$. This provides us with a closed-form expression of our safety constraints \eqref{multi_step:eq:constr:finite_horizon_safety} that deterministically guarantees the safety of our system over an arbitrary finite horizon $T$, \textit{given that} the system is contained in the sequence of ellipsoids $\mathcal{R}_t, t=0,..,T$. Hence, these constraints are as reliable as our multi-step ahead prediction technique and, consequently, as reliable as our statistical model.

% !TEX root = ../root.tex

\section{Safe Model Predictive Control}
\label{safempc}
Based on the previous results, we formulate a MPC scheme that is guaranteed to satisfy the safety condition in \cref{ps:def:epsilon_safety}:
%\begin{equation}
\begin{subequations}
\label{main:safempc:mpc:mpc_problem}
\begin{align}
&\underset{u_0,..,u_{T-1}}{\text{minimize}}
& & J_t(\mathcal{R}_0,..,\mathcal{R}_T) \\
\label{main:safempc:mpc:mpc_problem:a}
& \text{subject to }
 & & \mathcal{R}_{t+1} = \tilde{m}(\mathcal{R}_t,u_t), \, t=0,..,T-1   \\
& & & \mathcal{R}_t \subset \mathcal{X}, \, t = 1,..,T-1 \\
& & & u_t(\mathcal{R}_t) \subset \mathcal{U}, t = 0,..,T-1 \\
\label{main:safempc:mpc:mpc_problem:d}
& & & \mathcal{R}_T \subset \mathcal{X}_{\mathrm{safe}},
\end{align}
\end{subequations}
%\end{equation}
where $R_0 := \{x_t\}$ is the current state of the system and the intermediate state and control constraints are defined in \eqref{multi_step:ell:ell_polytope_state_constr}and \eqref{multi_step:ell:ell_polytope_control_constr}, respectively. The terminal set constraint $\mathcal{R}_T \subset \mathcal{X}_{\mathrm{safe}}$ has the same form as \eqref{multi_step:ell:ell_polytope_state_constr} and can be formulated accordingly. For now, we assume an arbitrary objective function and discuss how to choose $J_t$ to solve a RL task in \cref{safe_rl}.

Due to the terminal constraint $\mathcal{R}_T \subset \mathcal{X}_{\mathrm{safe}}$, a solution to  \eqref{main:safempc:mpc:mpc_problem} provides a sequence of feedback controllers $u_0,..,u_T$ that steers the system back to the safe set $\mathcal{X}_{\mathrm{safe}}$.
We cannot directly show that a solution to MPC problem \eqref{main:safempc:mpc:mpc_problem} exists at every time step (this property is known as recursive feasibility) without imposing additional assumptions. This is mainly due to the fact that we cannot (not even with high probability) guarantee that, after transitioning to the next time step and shifting our horizon $t \leftarrow t + 1$, and hence, $\mathcal{R}_{T-1} \leftarrow \mathcal{R}_T$, there exists a new state feedback controller $u_{T-1}$ such that $\tilde{m}(R_{T-1},u_{T-1}) = \mathcal{R}_T \subset \mathcal{X}_{\mathrm{safe}}, u_{T-1}(\mathcal{R}_{T-1}) \subset \mathcal{U}$. While it may be possible to enforce this in general, we would have to, among other difficulties, carefully deal with the nonlinearity and non-convexity of our MPC problem, as e.g. in \cite{Simon2013Nonlinear}, and the fact that our terminal set is not necessarily RCPI, which is a pre-requisite in many robust MPC approaches \cite{rawlings2009model}. \newline
However, employing a control scheme similar to standard robust MPC, we guarantee that such a sequence of feedback controllers exists at every time step as follows: Given a feasible solution $ \Pi_t = ( u_t^0,.., u_{t}^{T-1} )$ to \eqref{main:safempc:mpc:mpc_problem} at time $t$, we apply the first control $u_t^0$. In case we do not find a feasible solution to \eqref{main:safempc:mpc:mpc_problem} at the next time step, we shift the previous solution in a receding horizon fashion and append $\pi_{\mathrm{safe}}$ to the sequence to obtain $\Pi_{t+1} = ( u_t^1,.., u_{t}^{T-1}, \pi_{\mathrm{safe}} )$. We repeat this process until a new feasible solution exists that replaces the previous input sequence. This procedure is summarized in \cref{main:safempc:algo:safety_algorithm}.
We now state the main result of the paper that guarantees the safety of our system under the proposed algorithm.
\begin{thm}
\label{main:safempc:theorem:safe_algorithm}
Let $\pi$ be the controller defined through \cref{main:safempc:algo:safety_algorithm}, $x_0 \in \mathcal{X}_{\mathrm{safe}}$ and let the planning horizon $T$ be chosen as in \cref{prob:ass:reliable_statistical_model}.
Then the system \eqref{ps:eq:system} is $\delta-$safe under the controller $\pi$.
\end{thm}
\begin{proof}
From \cref{multi_step:lemma:t_step_delta_reach}, the ellipsoidal outer approximations (and by design of the MPC problem, also the constraints \eqref{ps:ass:state_constraits}) hold uniformly with high probability for \textit{all} closed-loop systems $f_\Pi$, where $\Pi$ is a feasible solution to \eqref{main:safempc:mpc:mpc_problem}, over the corresponding time horizon $T$ given in \cref{prob:ass:reliable_statistical_model}.
Hence we can show $\delta$-safety by induction.
Base case: If \eqref{main:safempc:mpc:mpc_problem} is infeasible, we are $\delta$-safe using the backup controller $\pi_{\mathrm{safe}}$ of \cref{ps:ass:safe_set}, since $x_0 \in \mathcal{X}_\mathrm{safe}$. Otherwise the controller returned from \eqref{main:safempc:mpc:mpc_problem} is $\delta$-safe as a consequence of \cref{multi_step:lemma:t_step_delta_reach} and the terminal set constraint that leads to $x_{t+T} \in \mathcal{X}_\mathrm{safe}$.
Induction step: let the previous controller be $\delta$-safe. At time step $t+1$, if \eqref{main:safempc:mpc:mpc_problem} is infeasible then $\Pi_t$ leads to a state $x_{t+T} \in \mathcal{X}_\mathrm{safe}$, from which the backup-controller is $\delta$-safe by \cref{ps:ass:safe_set}. If \eqref{main:safempc:mpc:mpc_problem} is feasible, then the return path is $\delta$-safe by \cref{multi_step:lemma:t_step_delta_reach}.
\end{proof}

\begin{algorithm}[t]
\caption{\textsc{SafeMPC}}
\label{main:safempc:algo:safety_algorithm}
\begin{algorithmic}[1]
\STATE \textbf{Input:} Safe policy $\pi_\mathrm{safe}$, dynamics model $h$, statistical model $(\mu_0,\Sigma_0)$. %\Comment{The g.c.d. of a and b}
%\For{$t=0,1,..$}
\STATE $\Pi_0 \gets \{\pi_{\mathrm{safe}},..,\pi_{\mathrm{safe}}\}$ with $|\Pi_0| = T$
\FOR{$t=0,1,..$}
	\STATE $J_t \gets$ objective from high-level planner
	\STATE feasible, $\Pi \gets$ solve MPC problem \eqref{main:safempc:mpc:mpc_problem}
    \IF{feasible} 
    \STATE $\Pi_t \gets \Pi$
    \ELSE
    \STATE $\Pi_t \gets \left(\Pi_{t-1,1:T-1},\pi_{\mathrm{safe}}  \right)$ 
    \ENDIF
    \STATE $x_{t+1} \gets$ apply $u_t = \Pi_{t,0}(x_t)$ to the system \eqref{ps:eq:system}
    \STATE $(\mu_{t+1}, \Sigma_{t+1}) \gets$ update statistical model with noisy  transition $(x_t, u_t, \tilde{f}(x_t,u_t)).$ 
\ENDFOR
%\EndProcedure
\end{algorithmic}%
\end{algorithm}

\begin{rem}
In \cref{main:safempc:algo:safety_algorithm}, we implicitly make the assumption that we have access to the true state of the system during planning, while updating the model with a noisy observation of the state. This is a standard assumption when acting in a Markov decision process (MDPs). The theoretical guarantees of \cref{main:safempc:algo:safety_algorithm} can be extended to noisy observations of the state by either assuming that the true state is contained in a compact set $B_0$ and setting $R_0 \supset B_0 $, or by jointly bounding the probability of the state to be contained in an ellipsoid around the current observation for all time steps.
\end{rem}
% !TEX root = ../root.tex

\section{MPC-based Safe Reinforcement Learning}
\label{safe_rl}
A sequence of ellipsoids computed with the proposed uncertainty propagation technique contains a trajectory of the system with high probability. While we can, by using this technique in combination with \cref{main:safempc:algo:safety_algorithm},  guarantee the safety of our system, we are ultimately interested in safely identifying our system through exploration and to solve a task in a safe reinforcement learning setting. To this end, we have to find answers to the following questions: Firstly, how can we best approximate and minimize the long-term cost \eqref{eq:back:model_based_rl:exp_cost} of our model predictive controller given our current knowledge of our system? And secondly, how can we gather \emph{task-relevant} observations from our system that help us in solving the given RL task?

\subsection{Safety and Performance}

We assume that we are given an analytic, twice differentiable cost-function $c(x_t,u_t)$, e.g. the squared distance to a desired set-point, for which we try to optimize the accumulated long term cost defined in equation~\eqref{eq:back:model_based_rl:exp_cost}. 
We could directly apply \cref{main:safempc:algo:safety_algorithm} to solve this RL task, e.g., by choosing the objective function $J_t = \sum_{t=0}^{T-1} \gamma^t c(p_t,k_t)$ for the MPC problem \eqref{main:safempc:mpc:mpc_problem}, where $p_t$ is the center of the ellipsoid $R_t =  \mathcal{E}(p_t,Q_t)$ in every step of the uncertainty propagation.

The issue with this approach is two-fold. Firstly, the proposed uncertainty propagation technique is designed to robustly outer-approximate the reachability of our system in order to guarantee safety. This stands in contrast to many of the techniques proposed in the literature which provide a \textit{stochastic} estimate of the state of the system in terms of probability distributions using our statistical model \cite{Deisenroth2011PILCO, Boedecker2014Approximate}. These estimation techniques give rise to more powerful stochastic planning approaches.
Secondly, we argue that a safety maneuver acts on a different (typically smaller) time scale than is necessary to estimate the long-term cost \eqref{eq:back:model_based_rl:exp_cost} of applying an action in the current time step.
Imagine a race car driving on a track. Estimating the impact of a steering maneuver before a turn on the velocity of the car after the turn typically requires a much longer look-ahead horizon than planning an emergency braking maneuver that decelerates the car towards a modest velocity.  
%Estimating the expected cost \eqref{eq:back:model_based_rl:exp_cost} based are potentially more accurate.
\subsection{Performance trajectory planning}
To avoid this issue, we propose to plan a probabilistic \textit{performance trajectory} $X_0,..,X_H$ under a sequence of inputs $u^{\textrm{perf}}_{0},..,u^{\textrm{perf}}_{H-1}$ using a stochastic performance-model $m_{\mathrm{perf}}$. The goal of the performance trajectory is to minimize an approximation of the cumulative cost \eqref{eq:back:model_based_rl:exp_cost},
\begin{align}
\label{mpc_rl:eq:approx_cum_cost}
    V^{u^{\textrm{perf}}_{0},..,u^{\textrm{perf}}_{H-1}}(x_0) &= \mathbb{E}[\sum_{t=0}^{H-1} \gamma^t c(X_t,u^{\textrm{perf}}_{t}(X_t)) | X_0 = x_0], \nonumber \\
    &\approx \sum_{t=0}^{\infty} \gamma^t c(x_t,u^{\textrm{perf}}_{t}(x_t)), 
\end{align}
where $x_{t+1} = f(x_t,u_t) $ is a trajectory of the real, \textit{deterministic} system \eqref{ps:eq:system} and $X_{t+1} \sim m_{\mathrm{perf}}(X_t,u^{\textrm{perf}}_{t}(X_t)), t=0,..,H-1$ is a sequence of random variables\footnote{While in our case we assume the system to be deterministic and our performance model to be stochastic, we can treat a deterministic performance model as a special case thereof.} that represents our incomplete, \textit{stochastic belief} about the future states of our system.
\subsection{Maintaining safety through simultaneous planning}
In order to maintain the safety of our system, we need to combine the performance trajectory with our provably safe MPC \cref{main:safempc:algo:safety_algorithm}. A simple solution would be given by the following two-stage process at every time-step:
\begin{enumerate}
    \item Find a sequence of performance controls $\hat{u}^{\textrm{perf}}_{0},..,\hat{u}^{\textrm{perf}} = \arg\min_{u^{\textrm{perf}}_{0},..,u^{\textrm{perf}}_{H-1}} \,  V^{u^{\textrm{perf}}_{0},..,u^{\textrm{perf}}_{H-1}}(x_0)$
    \item Solve the MPC problem \eqref{main:safempc:mpc:mpc_problem} with objective function $J_t = \sum_{i = 0}^{\min(H,T)} \textrm{dist}(u_i,\hat{u}^{\textrm{perf}}_{i})$,
\end{enumerate}
where $\textrm{dist}(\cdot,\cdot)$ is an appropriate distance function (e.g. the squared $l_2$ distance). By applying this scheme, e.g. similar to the one proposed in \cite{Wabersich2018Safe}, we could easily maintain the safety of our system. However, proximity of two actions $u_i,\hat{u}^{\textrm{perf}}_{i}$ in the action space does not necessarily imply similar performance, as we show in \cref{mpc_rl:example:failure_case}.
\begin{exmp}
\label{mpc_rl:example:failure_case}
Consider a simple system $f(x_t,u_t) = x_t + u_t, u_t \in \{-1,0,1\}, x \in \Z$ in a discrete state-action space without model error, i.e. $g \equiv 0$. Let the cost-function be given by $c(x) = 0, \, \forall x \in \Z\setminus\{-1,1\},\, c(-1) = -2,\,  c(1) = - 1$. We regard the case where our state constraints and safe set are given by $\mathcal{X} = \mathcal{X}_{\textrm{safe}} = \N$. Starting in $x_0 = 0$ and using planning horizons $T=H\geq1$, the optimal action in terms of performance would lead to the unsafe but low-cost state $x_1 = -1$, i.e. $u_0^{\textrm{perf}} = -1$. The closest safe action we can take is to apply $u_0 = 0$, resulting in the system being stuck in  $x_t = 0 \, \forall t \in \N$, although $u_0 = 1$ would directly lead to the best possible safe state.  
\end{exmp}
While the given example is highly simplified, it illustrates the fundamental issues that arise when a safety mechanism without an understanding of the task is combined with an RL agent that is ignoring the constraints. 
To approach this problem, we propose to \textit{simultaneously} plan the performance and safety trajectory. We maintain the safety of our system by coupling both trajectories, enforcing that $u^{\textrm{perf}}_{0} = u^{\textrm{safe}}_{0}$. That is, we apply the best possible action that is able to return the system to a safe state by means of a feasible safety trajectory.
This extended optimization problem is given by
\begin{equation}
\label{main:safempc:perf:mpc_problem}
\begin{aligned}
&\underset{\substack{u_0,..,u_{T-1} \\ u^{\mathrm{perf}}_{0},..,u^{\mathrm{perf}}_{H-1}}}{\text{minimize}}
& & V^{u^{\textrm{perf}}_{0},..,u^{\textrm{perf}}_{H-1}}(x_0) \\
& \text{subject to }
% & & R_{t+r+1} = \tilde{m}(R_{t+k},\pi_{t+k}), \, k=0,..,H-1
% \\
 & & \eqref{main:safempc:mpc:mpc_problem:a}-\eqref{main:safempc:mpc:mpc_problem:d},\,  t = 0,..,T-1 \\
& & & X_{t+1}  = m_{\mathrm{perf}}(X_t,u^{\mathrm{perf}}_{t}), t = 0,..,H-1 \\
& & & u_t = u^{\mathrm{perf}}_{t}, \, t = 0,..,r-1,
\end{aligned}
\end{equation}
with $r \geq 1$, possibly allowing even more than the first planned input to be identical.

\subsection{Exploration versus Exploitation}
A solution to the MPC problem \eqref{main:safempc:perf:mpc_problem} provides us with a control input that minimizes an \textit{estimate} \eqref{mpc_rl:eq:approx_cum_cost} of the long-term consequences of applying this input to the real system. In order to improve this estimate, we need a better approximation of our system based on observations that are used to update our statistical model. However, by just following a control policy that maximizes this estimated objective, we may never obtain crucial task-relevant information about our system in previously unseen regions of the state space. Hence, finding a good exploration strategy is crucial for many real-world RL applications.

Since our statistical model \textit{knows what it does not know} through the predictive variance $\Sigma_n$, we get a good estimate on which regions of the state space need to be explored. Consequently, a wide range of trajectory-based RL approaches and their corresponding exploration techniques attempt to find an efficient trade-off between exploration and exploitation by combining an estimate of the expected cost \eqref{mpc_rl:eq:approx_cum_cost} and the confidence intervals in the objective function, e.g.~\cite{Boedecker2014Approximate,Xie2016Modelbased,Kamthe2017DataEfficient}. We can incorporate any of these exploration techniques in the objective function of our extended MPC problem~\eqref{main:safempc:perf:mpc_problem}.

\begin{algorithm}[t]
\caption{Episodic Safe RL}
\label{mpc_rl:algo:safe_rl_episode}
\begin{algorithmic}[1]
\STATE\textbf{Input:} Safe policy $\pi_\mathrm{safe}$, dynamics model $h$, statistical model $(\mu_0, \Sigma_0)$%, initial observations $(Z_0, \hat{f}_0)$. %\Comment{The g.c.d. of a and b}
%\STATE $(\mu,\Sigma) \leftarrow$ update $(mu,\Sigma)$ with initial data $(Z_0, \hat{f}_0)$
\FOR{$i=0,1,...n_{ep}$}
    \STATE $x_0 \leftarrow$  Reset system to initial state inside $\mathcal{X}_{\mathrm{safe}}$
	\STATE $\mathcal{Z}_i, y_i \leftarrow $ Apply \cref{main:safempc:algo:safety_algorithm} with MPC problem \eqref{main:safempc:perf:mpc_problem}  and objective \eqref{mpc_rl:sat_cost} for $n_{steps}$ time steps without updating the statistical model.
	\STATE $(\mu,\Sigma) \leftarrow$ update $(\mu,\Sigma)$ with observations $(\mathcal{Z}_i, y_i)$
\ENDFOR
%\EndProcedure
\end{algorithmic}%
\end{algorithm}

% !TEX root = ../root.tex

\section{Practical Considerations}
\label{practice}
\cref{main:safempc:algo:safety_algorithm} theoretically guarantees that the system remains safe, while actively optimizing for performance via the MPC problem \eqref{main:safempc:perf:mpc_problem}. This problem can be solved by commonly used, nonlinear programming (NLP) solvers, such as the \textit{Interior Point OPTimizer (Ipopt, \cite{Wachter2006Implementation})}.
We consider possible design choices that could improve the performance in a practical application.

\subsection{Optimizing over affine feedback policies}

In practice, the affine feedback control structure introduced in \cref{safempc:affine_feedback} improves performance significantly. However, optimizing over $K_t$ in \eqref{main:multi_sep:eq:feedback_ctrl} seems to be challenging, both in terms of numerical stability and computational complexity. Hence, we pre-specify the feedback terms in all of our experiments and only optimize over the feed-forward terms.

\subsection{Lipschitz constants and eigenvalue computations}

In our multi-step ahead predictions \eqref{main:multi_step:eq:one_step_ahead_set}, we need to solve a generalized eigenvalue problem for every step in the planning horizon. We use the \textit{inverse power iteration}, an iterative method that asymptotically converges to the largest generalized eigenvalue of a pair of matrices \cite{Golub2012Matrix}. We run the algorithm for a fixed number of $p^2$ iterations to solve these intermediate eigenvalue problems. In practice, this seems to result in sufficiently accurate estimations.  

\begin{rem}
Due to the generalized eigenvalue problem, the uncertainty propagation \eqref{main:multi_step:eq:one_step_ahead_set} is not analytic.
However, we can still obtain exact derivative information by means of algorithmic differentiation, that is provided in many state-of-the-art optimization software libraries~\cite{Andersson2013b}.
\end{rem}

% !TEX root = ../root.tex

\section{Experiments}
\label{exp}

In this section, we evaluate the proposed safe MPC \cref{main:safempc:algo:safety_algorithm} to safely explore the dynamics of an inverted pendulum system and the corresponding episode-based RL \cref{mpc_rl:algo:safe_rl_episode} to solve a task in a cart-pole system with safety constraints. We provide the code to run all experiments detailed in this section on Github \footnote{\url{https://github.com/befelix/safe-exploration}.}

For our experiments we consider an episodic setting, where we interact with the system for $n_{steps}$ time steps, e.g. the time it ideally takes for an autonomous race car to finish a lap, and then use the observations gathered from the system after each rollout to update our statistical model. This procedure is repeated $n_{ep}$ times. We can directly replace \eqref{main:safempc:mpc:mpc_problem} with our new MPC problem \eqref{main:safempc:perf:mpc_problem} in \cref{main:safempc:algo:safety_algorithm} and apply it to our system during each rollout. This procedure is depicted in \cref{mpc_rl:algo:safe_rl_episode}.

The difference between the episodic setting and the MPC \cref{main:safempc:algo:safety_algorithm} lies mainly in the choice of the objective. However, a subtle but theoretically important difference is that we only update our model every $n_{steps}$ time steps. Hence, we technically require that the scaling factor $\beta$ for our statistical model holds for sample sizes $T \cdot n_{steps}$ in \cref{prob:ass:reliable_statistical_model}. Then, the safety guarantees of \cref{main:safempc:theorem:safe_algorithm} are maintained throughout our learning process in \cref{mpc_rl:algo:safe_rl_episode} by noting that feasibility of the problem \eqref{main:safempc:perf:mpc_problem} still guarantees existence of a safe return strategy and we can still use the same fall back strategy as before in case of infeasibility. We note that, as long as we account for a possible delay in updating our statistical model, our approach is not limited to episodic RL tasks only.

We use a GP with a mixture of linear and Mat\'ern kernels for both systems. For the inverted pendulum system, we initially train our model with a dataset $(\mathcal{Z}_0,\tilde{y}_0)$ sampled inside the safe set using the backup controller $\pi_{\mathrm{Safe}}$. That is, we gather $n_0 = 25$ initial samples $\mathcal{Z}_0 = \{ z_1^0,..,z^0_{n_0} \}$ with $z_i^0 = (x_i, \pi_{\mathrm{safe}}(x_i)), x_i \in \mathcal{X}_{\mathrm{safe}}, \, i = 1,..,n$ and observed next states $ \hat{f}_0 =\{\hat{f}^1_0,..,\hat{f}^{n_0}_0\}$. In the cart-pole RL task no prior observations from the system are available. The theoretical choice of the scaling parameter $\beta_{n,T}$ for the confidence intervals in \eqref{gp:lemma:rkhs_confidence_ivals} can be conservative and we choose a fixed value of $\beta_{n,T} = 2$ instead, following~\cite{Berkenkamp2017Safe}. For improved numerical stability and computational efficiency, we limit the number of training points used to update the GP to $150$ and use the \textit{maximum variance} selection procedure (see e.g. \cite{Jain2018Learning}) to sub-select the most informative samples in cases where more samples are available, however there exist more sophisticated, provably near-optimal selection procedures \cite{Krause2008NearOptimal}.

\begin{figure*}[t]
\includegraphics{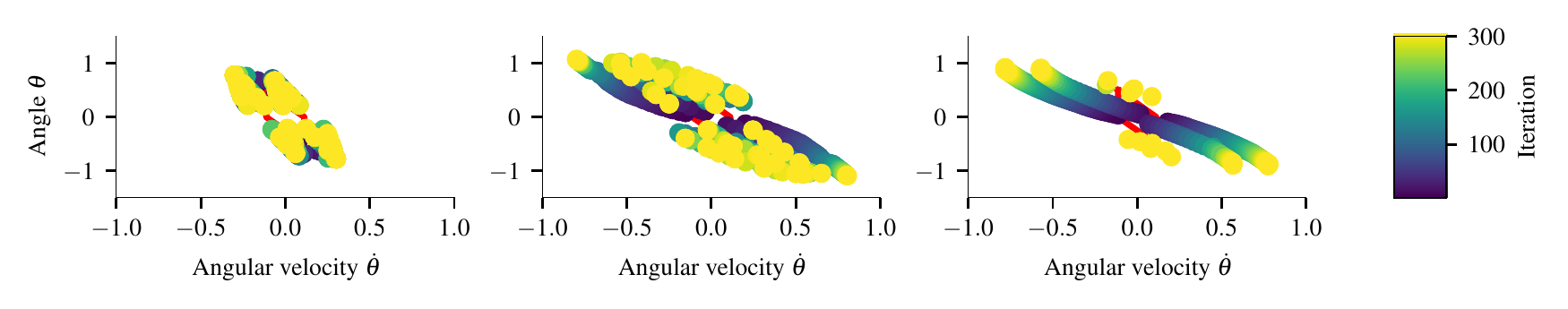}
\caption{Visualization of the samples acquired in the static exploration setting in~\cref{ssec:static_exploration} for $T \in \{1,4,5\}$. The algorithm plans informative paths to the safe set $\mathcal{X}_{\mathrm{safe}}$ (red polytope in the center). The baseline sample set for $T=1$ (left) is dense around origin of the system. For $T=4$ (center) we get the optimal trade-off between cautiousness due to a long horizon and limited length of the return trajectory due to a short horizon. The exploration for $T=5$ (right) is too cautious, since the propagated uncertainty at the final state is too large.}
\label{exp:fig:static_exploration:sampleset}
\end{figure*}

For experiments that employ a performance trajectory $m_{perf}$ as introduced in \cref{safe_rl}, we choose a common uncertainty propagation technique, that approximates the intractable posterior of the GP statistical model with uncertain, Gaussian inputs  by another Gaussian. That is, given a Gaussian distributed probabilistic belief of our system $X_t \sim \mathcal{N}(m_t,S_t)$, we approximate the intractable integral 
\begin{equation}
    p(\cdot|m_{t+1},S_{t+1}) = \mathbb{E}_{x_t \sim X_t}[ p_{\mathcal{GP}}(\cdot|x_t,u_t)]
\end{equation}
with the Gaussian $X_{t+1} \sim \mathcal{N}(m_{t+1},S_{t+1})$, where $p_{\mathcal{GP}}(\cdot|x_t,u_t)$ is the density function of the GP posterior for a deterministic input \eqref{found:gp:pred_mu}, \eqref{found:gp:pred_var}. Under a sequence of control inputs\footnote{Many approaches in the literature, including the one we use, also allow for linear state-feedback controllers as inputs.}, this results in a sequence of Gaussian random variables as a Bayesian approximation of a trajectory of the true system. We choose a technique where $m_{t+1}$ and $S_{t+1}$ are given by a Taylor approximation of the first and second moment of the posterior, respectively \cite{Girard2003Multiplestep}, but we note that a number approaches exist in this category that could be readily used. 

We try to actively encourage exploration by employing the \emph{saturating cost function} $c_{sc}(x) = 1 - \exp(-\frac{1}{2}(x-x_g)^{\mathrm{T}}W(x-x_g))$ as it behaves similarly to a squared distance function when close to the target $x_g$ and prefers uncertain states when being far away from the target state \cite{Deisenroth2011PILCO}.
Since our stochastic performance trajectory is given by a sequence of Gaussians, we can compute the corresponding surrogate objective,
\begin{equation}
\label{mpc_rl:sat_cost}
    \mathbb{E}[\sum_{t=0}^{H-1} \gamma^t c_{sc}(X_t,x_g) | X_0 = x_0], 
\end{equation}
in closed form. We can now replace \eqref{mpc_rl:eq:approx_cum_cost} with this objective function in our MPC formulation. 

\subsection{Inverted pendulum exploration}
\label{ssec:static_exploration}
As an initial experiment, we aim to safely identify a partially unknown inverted pendulum system through exploration, without trying to solve a specific task. To do so, we attempt to iteratively collect the most informative samples of the system, while preserving its safety. To evaluate the exploration performance, we use the mutual information $I(g_{\mathcal{Z}_n},g)$ between the collected samples $\mathcal{Z}_n = \{z_1,..,z_n\} \cup \mathcal{Z}_0$ and the GP prior on the unknown model-error~$g$, which can be computed in closed-form~\cite{Srinivas2010Gaussian}.

The continuous-time dynamics of the pendulum are given by $ml^2\ddot{\theta} = gml\sin(\theta) - \eta \dot{\theta} + u$, where $m = 0.15 \unit{kg}$ and $l=0.5\unit{m}$ are the mass and length of the pendulum, respectively, $\eta = 0.1 \unitfrac{Nms}{rad}$ is a friction parameter, and $g = 9.81 \unitfrac{m}{s^2}$ is the gravitational constant. The state of the system $x = (\theta,\dot{\theta})$ consists of the angle $\theta$ and angular velocity $\dot{\theta}$ of the pendulum. The system is underactuated with control constraints $\mathcal{U} = \{ u \in \R |-1 \leq u \leq 1\}$. Due to these limits, the pendulum becomes unstable and falls down beyond a certain angle. The origin of the system corresponds to the pendulum standing upright. The prior model $h$ is given by the linearized and discretized system around the origin, albeit with friction neglected the mass the pendulum being lower than for the true system as in ~\cite{Berkenkamp2017Safe}. The safety controller~$\pi_{\mathrm{safe}}$ is a discrete-time, infinite horizon linear quadratic regulator (LQR,\cite{Kwakernaak1972Linear}) of the \textit{true} system $f$ linearized and discretized around the origin with cost matrices
$ Q = \mathrm{diag}([1,2])$, $R = 20$. The corresponding safety region $\mathcal{X}_{\mathrm{Safe}}$ is given by a conservative polytopic inner-approximation of the true region of attraction of $\pi_{\mathrm{safe}}$. We do not impose state constraints, i.e.~$\mathcal{X} = \R^2$. However the terminal set constraint \eqref{main:safempc:mpc:mpc_problem:d} of the MPC problem \eqref{main:safempc:mpc:mpc_problem} acts as a stability constraint and prevents the pendulum from falling.
\subsubsection{Static Exploration}

For a first experiment, we assume that the system is \textit{static}, so that we can reset the system to an arbitrary state $x_n \in \R^2$ in every iteration. In the static case and without terminal set constraints, a provably close-to-optimal exploration strategy is to, at each iteration~$n$, select state-action pair~$z_{n+1}$ with the largest predictive standard deviation~\cite{Srinivas2010Gaussian}
\begin{equation}
\label{exp:max_gp_confidence_intervals}
z_{n+1} = \underset{z \in \mathcal{X} \times \mathcal{U}}{\arg\max} \sum_{1 \leq j \leq p } \sigma_{n,j}(z),
\end{equation}
where $\sigma_{n,j}^2(\cdot)$ is the predictive variance~\cref{found:gp:pred_var} of the $j$th $\mathcal{GP}(0,k_j)$ at the $n$th iteration. Inspired by this, at each iteration we collect samples by solving the MPC problem \eqref{main:safempc:mpc:mpc_problem} with cost function~$J_n = -\sum_{j=1}^{p} \sigma_{n,j}(x_0,u_0)$, where we additionally optimize over the initial state $x_0 \in \mathcal{X}$. Hence, we visit high-uncertainty states, but only allow for state-action pairs $z_n$ that are part of a feasible return trajectory to the safe set $\mathcal{X}_{\mathrm{Safe}}$.

Since optimizing over the initial state is highly non-convex, we solve the problem iteratively with $25$ random initializations to obtain a good approximation of the global minimizer. After every iteration, we update the sample set $\mathcal{Z}_{n+1} = \mathcal{Z}_n \cup \{z_n\}$, collect an observation $(z_n,\hat{f}_n)$ and update the GP models. We apply this procedure for varying horizon lengths.

The resulting sample sets are visualized for varying horizon lengths $T \in \{1,..,5\}$ with $300$ iterations in \cref{exp:fig:static_exploration:sampleset}, while \cref{exp:fig:static_exploration:inf_gain} shows how the mutual information of the sample sets $\mathcal{Z}_i,\, i = 0,..,n$ for the different values of $T$. For short time horizons ($T=1$), the algorithm can only slowly explore, since it can only move one step outside of the safe set. This is also reflected in the mutual information gained, which levels off quickly. For a horizon length of $T=4$, the algorithm is able to explore a larger part of the state-space, which means that more information is gained. For larger horizons, the predictive uncertainty of the final state is too large to explore effectively, which slows down exploration initially, when we do not have much information about our system. The results suggest that our approach could further benefit from adaptively choosing the horizon during operation, e.g. by employing a variable horizon MPC approach \cite{RichardsArthur2006Robust}, or by increasing the horizon when the mutual information saturates for the current horizon.
\begin{figure}[t]
\includegraphics{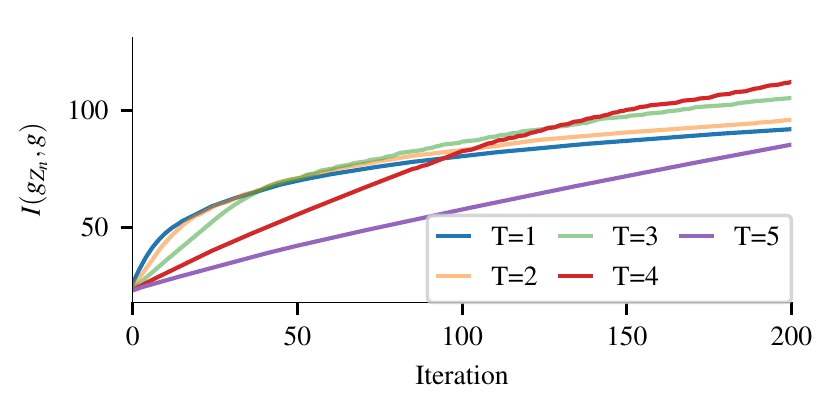}
\caption{Mutual information $I(g_{\mathcal{Z}_n},g), \, n=1,..,200$ for horizon lengths $T \in \{1,..,5\}$. Exploration settings with shorter horizon gather more informative samples at the beginning, but less informative samples in the long run. Longer horizon lengths result in less informative samples at the beginning, due to uncertainties being propagated over long horizons. However, after having gathered some knowledge they quickly outperform the smaller horizon settings. The best trade off is found for $T=4$.
}
\label{exp:fig:static_exploration:inf_gain}
\end{figure}

\subsubsection{Dynamic Exploration}

As a second experiment, we collect informative samples during operation; without resetting the system at every iteration. Starting at~$x_0 \in \mathcal{X}_\mathrm{safe}$, we apply \cref{main:safempc:algo:safety_algorithm} over $200$ iterations.
We consider two settings. In the first, we solve the MPC problem~\eqref{main:safempc:mpc:mpc_problem} with $-J_n$ given by \eqref{exp:max_gp_confidence_intervals}, similar to the previous experiments. In the second setting, we additionally plan a performance trajectory as proposed in \cref{safe_rl}. We define the cost-function $-J_t = \sum_{t=0}^{H} \mathrm{trace}(S_t^{1/2}) - \sum_{t=1}^{T} (m_t-p_t)^TQ_{\mathrm{perf}}(m_t-p_t)$, which maximizes the sum of predictive confidence intervals along the trajectory $m_1,..,m_H$, while penalizing deviation from the safety trajectory. We choose~$r = 1$ in the problem \eqref{main:safempc:perf:mpc_problem}, i.e. the first action of the safety trajectory and performance trajectory are the same. As in the static setting, we update our GP models after every iteration.

We evaluate both settings for varying $T \in \{1,..,5\}$ and fixed $H=5$ in terms of their mutual information in \cref{exp:fig:dynamic_expl:information_gain}. We observe a similar behavior as in the static exploration experiments and get the best exploration performance for $T=4$, with a slight degradation of performance for $T=5$ after 200 iterations. By comparing the exploration performance between iteration $50$ and $200$, we can see that influence of longer return trajectories on the exploration performance only comes into play after a certain number of iterations. This can be seen by comparing the similar performance of $T=3$ and $T=4$ after $50$ iterations with the significantly improved performance for $T=4$ after $200$ iterations. The setting $T=3$ during the same period only sees modest performance improvements. We can see that, except for $T=1$, the performance trajectory decomposition setting consistently outperforms the standard setting. Planning a performance trajectory (green) provides the algorithm with an additional degree of freedom, which leads to drastically improved exploration performance.

\begin{figure}[t]
\includegraphics{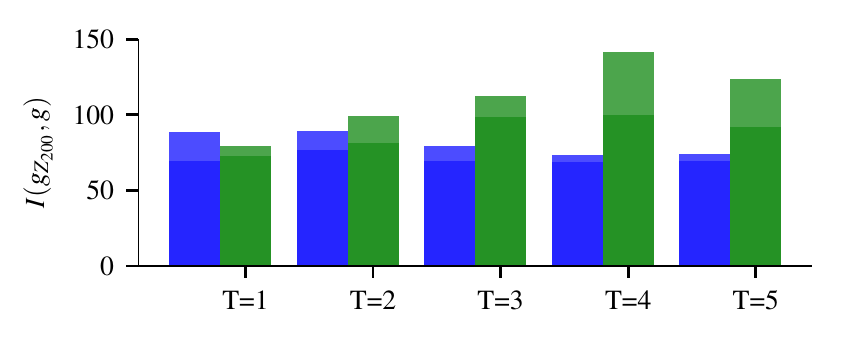}
\caption{Comparison of the information gathered from the system after $50$ (dark colors) and $200$ (light colors) iterations for the standard setting (blue) and the setting where we plan an additional performance trajectory (green). %The performance trajectory setting consistently outperforms the standard setting, except for $T=1$.
}
\label{exp:fig:dynamic_expl:information_gain}
\end{figure}

\subsection{Cart-pole safe reinforcement learning}
%We tackle a cart-pole swing-up RL benchmark task \cite{Deisenroth2011PILCO,Xie2016Modelbased}, where in the original task an underactuated cart learns to swing up a pendulum to an upright position.
We tackle a cart-pole balancing task, where and underactuated cart learns to balance a pendulum in upright position while moving along a rail from starting position to a goal state. We initialize the pendulum in an upright position, with an initial position of the cart at $x^{cart}_0 = -2$  and we want to drive the system to the goal-position at $x^{cart}_g = 2.6$. We limit the length of the rail by $x^{cart} \in [-10,3.0]$ and simulate a floor, the pendulum is not allowed to hit, i.e. $\theta \in [-90,90]$ degree, where $\theta$ is the angle of the pendulum.
 The equations of motion of the system are given by
\begin{align*}
    (M+m)\ddot{x}^{cart} - ml\ddot{\theta}\cos{\theta}+ ml\dot{\theta}^2\sin{\theta} = 0, \\
     ml\ddot{\theta}-m\dot{x}^{cart}\cos{\theta}  -mg\sin{\theta} = u - \eta\dot{x}^{cart},
\end{align*}
where the velocity of the cart and the angular velocity are given by $\dot{x}^{cart}$ and $\dot{\theta}$, respectively. The mass of the cart is $M = 0.5 \unit{kg}$ and the corresponding mass of the pendulum is given by $m = 0.5 \unit{kg}$. The gravitational constant is again given by $g = 9.81  \unitfrac{m}{s^2}$ and the friction coefficient of the rail is $\eta = 0.1$. The origin of the system describes the cart at rest with the pendulum in an upright position The cart is strictly limited in its capability to overshoot the goal position due to the constraint $x^{cart} < 3$. The angular constraints $\theta \in [-90,90] \deg$ require the system to remain stable while driving to the goal-position. The actuators of the cart are limited to $u \in [-5,5] $. Again, we use a LQR controller based on the linearized true system with cost-matrices $Q= \diag([4,8,12,2])$, $R = 40$  as the safety controller $\pi_{\mathrm{safe}}$.

We apply the episode-based safe RL \cref{mpc_rl:algo:safe_rl_episode} over $8$ episodes, where in each episode we interact with the system over $n_{steps} = 50$ time steps. We then measure the performance of a rollout $x_0,..,x_{n_{steps}}$ via the accumulated squared distance to the goal $ C_{ep} =  \sum_{t=0}^{n_{steps}} 0.1 (x^{cart}_t-x_g)^2$. All reported results are averaged over $6$ repetitions of the corresponding experiment.

We compare the performance of the SafeMPC algorithm with and without the probabilistic planning module introduced in \cref{safe_rl}. In the settings where we use a performance trajectory, we set the number of planning steps to $H=15$ and choose the expected squared exponential cost function \eqref{mpc_rl:sat_cost} as the performance objective.
In the case of no performance trajectory, i.e. $H=0$, we choose the objective $J_t = \sum_{t=0}^{T-1} c_{rl}(p^{cart}_t-x_g)^2$, where $p_0,..,p_{T-1}$ are the centers of the ellipsoids in the safety trajectory at each time step.

\subsubsection{Influence of the performance trajectory}

We report the performance of the last episode for $T \in \{1,2,3,4\}$ and $H \in \{0,15\}$ in \cref{exp:rl:fig:final_ep_std_perf}.
In \cref{exp:rl:fig:all_ep_std_perf}, we compare the summed cost $C_{ep}$ for each episode with $H = 15$ and $T \in \{1,2,3,4\}$.     

As for the exploration experiments, we can see that planning a performance trajectory greatly improves the performance of our system. In this RL task, the agent even fails to find any reasonable strategy when planning without using the additional performance trajectory to estimate the long-term utility of its actions. In the settings with $T \in \{1,2\}$, the performance seems to improve quite steadily over the course of episodes. Whereas, for $T \in \{3,4\}$, the performance seems to fluctuate more.  This could be due to the system switching to the safety trajectory more often, as longer safety trajectories allow more aggressive maneuvers from which the system then has to recover again. This seems to be harming overall performance. In future work, we want to investigate how to alleviate the negative effects of this switching control design. We note that, depending on the design of the RL problem, much longer horizons of the safety trajectory may be beneficial.

\begin{figure}[t]
\includegraphics{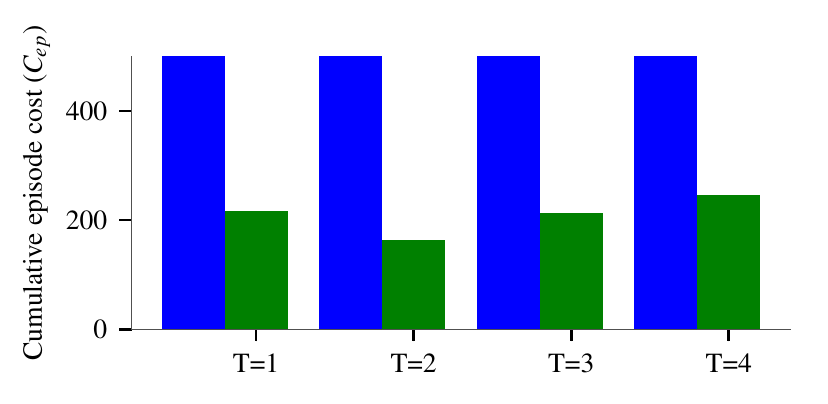}
\caption{Comparison of the performance of RL agents with varying $T \in \{1,2,3,4\}$ after eight episodes for the setting with (green) and without performance trajectory (blue).
}
\label{exp:rl:fig:final_ep_std_perf}
\end{figure}

\begin{figure}[t]
\includegraphics{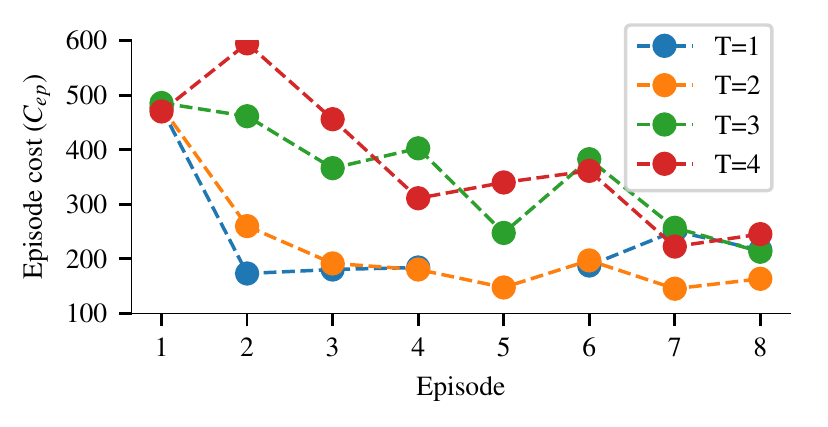}
\caption{Performance of the RL agents with safety trajectory length $T \in \{1,2,3,4\}$ and performance trajectory length $H=15$.%The performance trajectory setting consistently outperforms the standard setting, except for $T=1$.
}
\label{exp:rl:fig:all_ep_std_perf}
\end{figure}

%\subsubsection{Choice of the safety coefficient}
%Due to the choice of a fixed safety coefficient $\beta$, we investigate how different values thereof influence the performance and safety of our algorithm.
%We consider the values $\beta \in \{1.0,2.0,3.0\}$ in the performance setting with $H=10$ and $T \in \{1,2,3,4\}$. We report the episode cost and percentage of successful rollouts in \cref{}. 
%\begin{table}[h]
%\label{exp:table:varying_beta}
%\begin{center}  
%\caption{Ratio of failed rollouts and accumulated cost (averaged over successful rollouts) for varying lengths of the safety trajectory and values of $\beta$ with fixed performance trajectory length $H = 10$. Lower is better for both benchmarks. \\[2pt]} 
%\begin{tabular}{l*{3}{|c}}\toprule[1.5pt]
%T & $\beta = 1.0$&$\beta = 2.0$&$\beta = 3.0$\\\hline
%1 & 0.9 / 55.4 &0.9 / 55.4& 0.9 / 55.4\\
%2 &0.9 / 55.4&0.9 / 55.4&0.9 / 55.4\\
%3 &0.9 / 55.4&0.9 / 55.4&0.9 / 55.4 \\
%\bottomrule
%\end{tabular}
%\\[10pt]
%\end{center}
%\label{tab:my_label}
%\end{table}
\subsubsection{Baseline comparison}
We compare our approach to a baseline, where we set $T=0$, i.e. we remove the safety trajectory and only consider the performance trajectory with the objective \eqref{mpc_rl:sat_cost}. We use probabilistic chance constraints to approximately enforce the state and control boundaries, as detailed in \cite{Hewing2017Cautious}. We consider different lengths of the performance trajectory $H \in \{5,10,15\}$ and report the percentage of rollouts that were successful and the accumulated cost for the different settings in \cref{exp:table:performance_only}. We can see that the system crashes in a large number of rollouts in all three settings, where for $H=5$ the system crashes almost constantly. There also seem to be two types of shortcomings that lead to safety failures. For short horizons ($H=5$) the problem seems to lie in the fact that we do not account for the long-term effect on safety when taking an action, whereas for long horizons ($H=15$), finding a feasible solution to the MPC problem seems to be an issue, leading to safety failures due to the missing fallback strategy.

We note that the computational complexity of both approaches is dominated by the time to compute the GP predictive uncertainty \eqref{found:gp:pred_var} along the planned trajectories in each iteration of the solver. Hence, when carefully optimized for runtime, it should be indeed possible to execute our proposed algorithm in real-time \cite{Hewing2017Cautious, Boedecker2014Approximate}.  
\begin{table}[t]
\begin{center}  
\label{exp:table:performance_only}
\caption{Ratio of failed rollouts of all episodes and cumulative final episode cost (averaged over successful rollouts) for varying lengths $H \in \{5,10,15\}$ of the performance trajectory. Lower is better for both benchmarks.} 
\begin{tabular}{c||cc|cc}%\toprule[1.5pt]
 & \multicolumn{2}{c}{\underline{Cautious MPC}} & \multicolumn{2}{c}{\underline{SafeMPC ($T=2$)}} \\
$H$&Failures[\%] & $C_{ep}$ & Failures[\%] & $C_{ep}$\\\hline
5 & 87.5 & 281.88 & 0.0 & $> 1000$ \\
10& 10.4 & 164.26 & 0.0 & 661.04 \\
15& 18.7 & 153.16 & 0.0 & 163.42 \\
%\bottomrule
\end{tabular}
\\[10pt]
\label{tab:my_label}
\end{center}
\end{table}

% !TEX root = ../root.tex

\section{Conclusion}
\label{discuss}
We developed a safe learning-based MPC framework that can solve reinforcement learning tasks under state and control constraints. We proposed two reliable uncertainty propagation techniques that underlie the algorithm. By combining the safety features of the learning-based MPC framework with techniques from model-based RL, we can guarantee the safety of the system while learning a given task. We experimentally showed that our proposed RL algorithm is capable of learning a task in a simulated cart-pole system without falling over or overshooting the end of the rail. 
\ifCLASSOPTIONcaptionsoff
  \newpage
\fi



%\bibliographystyle{IEEEtran}
\bibliographystyle{plain}        % Include this if you use bibtex 
\bibliography{autosam} 

\end{document}